\documentclass[preprint,12pt,authoryear,nopreprintline]{elsarticle}

\usepackage{amssymb}
\usepackage{amsmath}
\usepackage{graphicx}%
\usepackage{multirow}%
\usepackage{amsmath,amssymb,amsfonts}%
\usepackage{amsthm}%
\usepackage{mathrsfs}%
\usepackage[title]{appendix}%
\usepackage{xcolor}%
\usepackage{textcomp}%
\usepackage{manyfoot}%
\usepackage{booktabs}%
\usepackage{algorithm}%
\usepackage{algorithmicx}%
\usepackage{algpseudocode}%
\usepackage{listings}%

\usepackage{caption}
\usepackage{subcaption}
\usepackage{bm}
\usepackage[hidelinks]{hyperref}
\newtheorem{theorem}{Theorem}
\newtheorem{lemma}{Lemma}

\newtheorem{example}{Example}
\newtheorem{definition}{Definition}

\begin{document}

\begin{frontmatter}



\title{Characterization and Computation of Normal-Form Proper Equilibria in Extensive-Form Games via the Sequence-Form Representation} 


\author[1,3]{Yuqing Hou}\ead{yuqinghou2-c@my.cityu.edu.hk}

\author[2]{Yiyin Cao\corref{cor1}}\ead{yiyincao2-c@my.cityu.edu.hk}

\author[3]{Chuangyin Dang}\ead{mecdang@cityu.edu.hk}

\cortext[cor1]{Corresponding author}

\affiliation[1]{organization={Department of Automation, University of Science and Technology of China},
	city={Hefei, Anhui},
	postcode={230026}, 
	country={China}}

\affiliation[2]{organization={School of Management, Xi'an Jiaotong University},
	city={Xi'an, Shanxi},
	postcode={710049}, 
	country={China}}

\affiliation[3]{organization={Department of Systems Engineering, City University of Hong Kong},
	city={Hong Kong},
	country={China}}

\begin{abstract}
Normal-form proper equilibrium, introduced by Myerson as a refinement of normal-form perfect equilibrium, plays an important role in the equilibrium analysis of extensive-form games because of its stronger stability properties. However, the exponential growth of the normal-form representation makes existing approaches for directly computing normal-form proper equilibria computationally impractical. To address this challenge, we introduce sequence-form proper equilibrium by defining a maximal payoff function over sequences, and prove its strategic equivalence to normal-form proper equilibrium. For computation, we further derive an equivalent formulation of sequence-form proper equilibrium via a class of perturbed games in sequence form, in which the perturbed realization-plan space is induced by imposing probability lower bounds on restricted sequence sets to encode the desired perturbation hierarchy. Based on this formulation, we develop two differentiable path-following methods for computing normal-form proper equilibria. These methods rely on artificial sequence-form games in which the expected payoff functions are regularized by logarithmic or entropy terms through an auxiliary variable. For each method, as the auxiliary variable tends to zero, the limit point of the induced smooth equilibrium path can derive one or more normal-form proper equilibria of the original game. Finally, Numerical experiments demonstrate the effectiveness and efficiency of the proposed methods.
\end{abstract}



\end{frontmatter}



\section{Introduction}
Extensive-form games provide a foundational framework for modeling noncooperative strategic interactions that unfold sequentially. Although Nash equilibrium remains the most widely used solution concept owing to its relative computational tractability, an extensive-form game may admit multiple Nash equilibria, some of which deviate from intuitive predictions regarding off-equilibrium-path behavior and equilibrium outcomes. Equilibrium refinements \citep{SeltenReexaminationperfectnessconcept1975,MyersonRefinementsNashequilibrium1978} mitigate these limitations by imposing stronger rationality requirements. Among these refinements, normal-form proper equilibrium is especially attractive because it eliminates many Nash equilibria while inducing sequentially rational outcomes, thereby explaining \cite{KohlbergStrategicStabilityEquilibria1986}'s view that the reduced normal form preserves all essential decision-relevant information. However, its normal-form definition entails an exponentially increasing pure-strategy space, making direct computation prohibitive even for moderately sized games and therefore motivating the search for a compact characterization.

\subsection{Sequence-Form Characterization}

The sequence form introduced by~\cite{romanovskiiReductionGameComplete1962,Kollercomplexitytwopersonzerosum1992,vonStengelEfficientComputationBehavior1996} has proved highly effective for the characterization of Nash equilibria and certain equilibrium refinements in extensive-form games, thereby exhibiting significant advantages in complexity analysis and equilibrium computation. In the sequence form, pure strategies are replaced by action sequences, and realization plans assign sequence weights under recursively defined linear constraints, thereby producing a representation that grows linearly with the size of the game. Building on this representation, von Stengel showed that two-player Nash equilibria can be characterized as solutions to linear optimization or complementarity problems, whereas $n$-player Nash equilibria require computationally challenging nonlinear formulations. In response to this computational obstacle, \cite{houSequenceformDifferentiablePathfollowing2025} developed a fixed-point formulation based on the best-response problem, leading to a polynomial system whose solutions encode all Nash equilibria of the game.

In light of the structural properties of sequence form and the defining principles of equilibrium refinements, several studies have explored sequence-form characterizations of certain refinements by introducing appropriately designed perturbations into the realization-plan space. \cite{MiltersenComputingquasiperfectequilibrium2010} demonstrated that, in two-player games, the limit points of Nash equilibria in a particular class of perturbed games in the sequence form coincide with a subset of the quasi-perfect equilibria of the original extensive-form game. This line of inquiry was subsequently generalized by \cite{gattiCharacterizationQuasiperfectEquilibria2020}, who provided a complete characterization of quasi-perfect equilibria for $n$-player games. In a related direction, \cite{hansenComputationalComplexityComputing2021a} introduced the generalized $\varepsilon$-permutahedron as a perturbation device and provided a sequence-form characterization of a subset of quasi-proper equilibria. The perturbations used to characterize quasi-perfect and quasi-proper equilibria are defined locally at each information set to capture the sequential-rationality considerations underlying these refinements. From a global perspective, \cite{houSequenceFormCharacterizationDifferentiable2025} proposed a perturbation scheme in which terminal-sequence perturbations recursively determine those of preceding sequences through the structure of realization plans. Nevertheless, the above perturbation constructions cannot be directly extended to normal-form proper equilibrium, which requires perturbation magnitudes reflecting relative payoffs among all pure strategies.

\subsection{Existing Computational Methods}

Path-following methods have been widely used for the exact computation of equilibria and naturally provide a constructive proof of equilibrium existence. Foundational insights into these methods come from the literature on Nash equilibrium computation in normal-form games. \cite{LemkeEquilibriumPointsBimatrix1964} proposed a complementarity pivoting algorithm for bimatrix games, which was later extended to $n$-player games by \cite{RosenmullerGeneralizationLemkeHowson1971} and \cite{WilsonComputingEquilibriaNPerson1971}. However, these extensions provided little guidance for practical computation. To address this limitation, \cite{GarciaSimplicialApproximationEquilibrium1973}, \cite{vanderLaanComputationFixedPoints1982}, \cite{Doupnewsimplicialvariable1987}, and \cite{HeringsComputationNashEquilibrium2002} developed simplicial path-following methods, which remain applicable even when the underlying system is non-smooth. However, as noted by \cite{allgowerPiecewiseLinearMethods2000a}, their practical efficiency is often limited for large-scale problems, as step-length adaptation and the exploitation of problem structure are difficult to implement. This motivated the development of differentiable path-following methods that exploit the intrinsic smoothness of Nash equilibrium systems. Important contributions include \cite{Heringsdifferentiablehomotopycompute2001} and \cite{GovindanglobalNewtonmethod2003}, which significantly improve the efficiency and convergence of equilibrium computation.

Path-following methods are well-suited for computing equilibrium refinements in normal-form games, as these refinements often arise as limits along equilibrium paths generated by perturbed games. For perfect equilibrium, \cite{vandenelzenProcedureFindingNash1991} introduced a complementary pivoting algorithm for bimatrix games, while \cite{chenReformulationBasedSimplicialHomotopy2019} proposed a simplicial homotopy method for $n$-player games. Subsequently, \cite{Chendifferentiablehomotopymethod2021} introduced an exterior-point differentiable path-following method for approximating perfect equilibria based on a convex-quadratic-penalty game. By exploiting the selection properties inherent in Nash's mappings, Harsanyi's tracing procedures, and logistic quantal response equilibrium, \cite{caoComplementarityEnhancedNashs2022}, \cite{CaovariantHarsanyitracing2022}, and \cite{caoVariantLogisticQuantal2024} developed their respective variants and distinct differentiable path-following methods to select an exact perfect equilibrium. Comparatively few methods exist for proper equilibria because of their more intricate perturbation structure. Following Kohlberg and Mertens' construction of perturbed games based on permutation-dependent strategy spaces, \cite{yamamotoPathfollowingProcedureFind1993} proposed a piecewise-differentiable path-following method. \cite{vanderlaanExistenceApproximationRobust1999} developed a simplicial path-following method to approximate robust stationary points, which coincide with proper equilibria in normal-form games. \cite{sorensenComputingProperEquilibrium2012} reformulated the computation of proper equilibria in bimatrix games as a polynomial-size linear complementarity problem and solved it using Lemke's (\citeyear{LemkeBimatrixEquilibriumPoints1965}) method. More recently, \cite{caoDifferentiablePathFollowingMethod2023} introduced a differentiable path-following method based on a compactly represented square-root-barrier game for computing proper equilibria and provided numerical results supporting its computational efficiency.

Research on path-following methods within the sequence-form representation has demonstrated favorable performance in equilibrium computation for extensive-form games, primarily in two-player settings. \cite{KollerEfficientComputationEquilibria1996} applied Lemke's (\citeyear{LemkeBimatrixEquilibriumPoints1965}) method to the sequence-form linear complementarity problem for computing Nash equilibria, with practical efficiency demonstrated by \cite{KollerRepresentationssolutionsgametheoretic1997}. Subsequently, \cite{vonStengelComputingNormalForm2002} adapted van den Elzen and Talman's (\citeyear{vandenelzenProcedureFindingNash1991}) method from normal-form games to the sequence form, providing a way to compute normal-form perfect equilibria in extensive-form games. \cite{MiltersenComputingquasiperfectequilibrium2010} further extended Koller et al.'s (\citeyear{KollerEfficientComputationEquilibria1996}) algorithm to handle perturbed games, thereby enabling the computation of quasi-perfect equilibria. Research on $n$-player games remains limited. Exploiting the sequence-form representation, Govindan and Wilson~\cite{govindanStructureTheoremsGame2002} established a structure theorem for perturbed extensive-form games and introduced a piecewise differentiable path-following method for computing Nash equilibria. Building on sequence-form characterizations of Nash equilibria and normal-form perfect equilibria, \cite{houSequenceformDifferentiablePathfollowing2025,houSequenceFormCharacterizationDifferentiable2025} developed corresponding differentiable path-following methods to compute them.

\subsection{Challenges and Contributions}

This paper studies the sequence-form characterization and computation of normal-form proper equilibria in finite $n$-player extensive-form games with perfect recall. This objective raises several challenges. Unlike pure strategies, sequences do not directly define independent probability measures; their realization probabilities are linked by recursive linear constraints, which makes sequence-level payoff comparisons nontrivial. Moreover, normal-form proper equilibrium is defined through pairwise comparisons among pure strategies in the exponentially large mixed-strategy space, and this comparison structure is difficult to preserve under the compact realization-plan representation. These difficulties also complicate the construction of path-following methods, which require a tractable starting point and a smooth path that selects the intended refinement. Building on these observations, our contributions are threefold:

\begin{enumerate}
	\item We develop a sequence-based payoff-comparison framework by defining a maximal payoff function for sequences. This framework yields the notion of sequence-form proper equilibrium, which we prove to be strategically equivalent to normal-form proper equilibrium.\footnote{A related result was obtained independently in contemporaneous work~\cite{zhangComplexityProperEquilibrium2026}, which was developed over the same period as our preprint~\cite{houCharacterizationComputationNormalForm2026}. The two works, however, are presented in different ways and intended for different purposes.} We further show that the framework also applies to characterizing several additional equilibrium concepts.
	
	\item We provide an equivalent formulation of sequence-form proper equilibrium through perturbed sequence-form games, in which realization probability lower bounds reflecting the prescribed perturbation hierarchy are imposed on restricted sequence sets. We further show that sequence-form proper equilibria arise as the limits of Nash equilibria of these perturbed games as the perturbation vanishes.
	
	\item We propose two differentiable path-following methods for computing normal-form proper equilibria by constructing artificial sequence-form games whose expected payoff functions are augmented with logarithmic or entropy terms through an auxiliary variable. Experimental results further validate the effectiveness and efficiency of our proposed methods.
\end{enumerate}

The remainder of this paper proceeds as follows. Section~\ref{nfpr-sec-prm1} reviews the relevant preliminaries. Section~\ref{nfpr-sec-prm2} introduces sequence-form proper equilibrium and proves its equivalence to normal-form proper equilibrium, while Section~\ref{nfpr-sec-prm3} provides an equivalent characterization through perturbed sequence-form games. Section~\ref{nfpr-sec-prm4} develops two differentiable path-following methods for computation. Numerical results are presented in Section~\ref{nfpr-sec-prm5}. Section~\ref{nfpr-sec-prm6} concludes the paper.

\section{Preliminaries}\label{nfpr-sec-prm1}
\begin{table}[htbp]
	\centering
	\caption{Notation for Extensive-Form Games}
	\label{nfpr-tab-pre1}
	\fontsize{10pt}{11pt}\selectfont
	\setlength{\tabcolsep}{3pt}
	\begin{tabular}{ll}
		\toprule
		Symbol & Explanation\\
		\midrule
		$N=\{1,2,\ldots,n\}$ & Set of players\\
		$N_c=N\cup\{c\}$ & Set of players and chance player $c$\\
		$a$ & Action taken by a player\\
		$H$ & Set of histories,  \\&$\emptyset\in H$ and $\langle a_1,\ldots,a_L\rangle\in H$ if $\langle a_1,\ldots,a_K\rangle\in H$ and $L<K$\\
		$Z$ & Set of terminal histories\\
		$A(h)=\{a\mid (h,a)\in H\}$ & Set of actions after a nonterminal history $h$\\
		$P(h)$ & Player who takes an action after $h$\\
		$f_{c}(a|h)$ & Probability that chance player $c$ takes action $a$ after $h$\\
		$-i$ & All non-chance players excluding player $i\in N$\\
		$\mathcal{I}_{i}$ & Collection of information partitions of $\{h\in H\mid P(h)=i\}$\\
		$M_{i}=\{1,\ldots,m_{i}\}$ & Set of information partition indices for player $i\in N_c$\\
		$I^{j}_{i}\in\mathcal{I}_{i},j\in M_{i}$ & $j$th information set of player $i\in N_c$, \\&$A(I^j_i)\triangleq A(h)= A(h')$ whenever $h,h'\in I^j_i$\\
		$\succsim_i$ & Preference relation of player $i\in N$ \\
		$u_z^{i}:Z\to\mathbb{R}$ & Payoff function of player $i\in N$\\
		$R_{i}(h)$ & Record of player $i\in N_c$'s experience along $h$\\
		$|C|$ & Cardinality of a finite set $C$\\
		$m_0=\sum_{i\in N}m_i$ & Number of information sets\\
		$n_0=\sum_{i\in N}\sum_{j\in M_i}|A(I^j_i)|$ & Number of actions for non-chance players\\
		\bottomrule
	\end{tabular}
\end{table}
Following \cite{OsborneCourseGameTheory1994}, an extensive-form game can be described as $\Gamma=\langle N, H, P, f_c, \{{\cal I}_i\}_{i\in N}, \{\succsim_i\}_{i\in N}\rangle$, with the notation summarized in Table~\ref{nfpr-tab-pre1}. In this paper, we restrict attention to finite extensive-form games with perfect recall. Here, ``finite'' means that the set $H$ is finite. Perfect recall requires that, for every player $i\in N_c$, any two histories $h$ and $h'$ contained in the same information set of $i$ satisfy $R_i(h)=R_i(h')$.

The normal-form representation of $\Gamma$ is written as $\Gamma_n=\langle N, S, \sigma^c, \{u^i\}_{i\in N}\rangle$ and the corresponding notation is explained in Table~\ref{nfpr-tab-pre2}. In particular, a pure strategy of player $i\in N_c$ is represented by a function $s^i$ that assigns to each information set $I^j_i$ an action in $A(I_i^j)$. In this paper, unless otherwise specified, we consider the reduced normal-form representation, where a pure strategy $s^i$ of player $i\in N_c$ specifies an action at an information set $I^j_i$ if and only if, for $\langle a_1,\ldots,a_L\rangle\in I^j_i$, $s^i(I^{j_q}_i)=a_q$ holds for all $1\leq q\leq L$ with $j_q\in M_i,a_q\in A(I^{j_q}_i)$. For computational convenience, we write $s^i(a)=1$ if $s^i(I_i^j)=a$, and $s^i(a)=0$ otherwise. Then, for any pure strategy profile \(s=(s^i:i\in N_c)\), the payoff of player \(i\in N\) is give by $u^i(s)=\sum_{h=\langle a_1,\ldots,a_L\rangle\in Z}u^i_z(h)\prod_{q=0}^{L-1}s^{P(\langle a_1,\ldots,a_q\rangle)}(a_{q+1})$. Given a mixed strategy profile $\sigma=(\sigma^i:i\in N)\in \Xi$, the expected payoff of player $i\in N$ is given by $u^i(\sigma)=\sum_{s^i\in S^i}\sigma^i(s^i)u^i(s^i,\sigma^{-i})$ with $u^i(s^i,\sigma^{-i})=\sum_{s^{-i}\in S^{-i}}u^i(s^i,s^{-i})\prod_{i_q\in N_c\backslash \{i\}}\sigma^{i_q}(s^{i_q})$. A mixed strategy profile $\sigma^*$ is referred to a normal-form Nash equilibrium if, for any $i\in N$ and $s^i\in S^i$, $\sigma^{*i}(s^i)=0$ whenever $u^i(s^i,\sigma^{*-i})< u^i(\tilde s^i,\sigma^{*-i})$ for some $\tilde s^i\in S^i$. \cite{SeltenReexaminationperfectnessconcept1975} introduced normal-form perfect equilibrium, which eliminates some unreasonable Nash equilibria, as defined below.
\begin{definition}\label{nfpe-def-pre1}
	Let $\Gamma$ be an extensive form game. For any sufficiently small $\varepsilon > 0$, a totally mixed strategy profile $\sigma(\varepsilon)\in \Xi$ is an $\varepsilon$-normal-form perfect equilibrium of $\Gamma$ if $\sigma^i(\varepsilon;s^i)\leq\varepsilon$ whenever $u^i(s^i, \sigma^{-i}(\varepsilon)) < u^i(\tilde s^{i},\sigma^{-i}(\varepsilon))$ for all $i\in N$ and $s^i,\tilde s^i\in S^i$. A mixed strategy profile $\sigma^*\in \Xi$ is defined as a normal-form perfect equilibrium of game $\Gamma$ if $\sigma^*$ is a limit point of some sequence $\{\sigma(\varepsilon^k)\}_{k=1}^\infty$, where $\lim_{k\to\infty}\varepsilon^k=0$ and each $\sigma(\varepsilon^k)$ is an $\varepsilon^k$-normal form perfect equilibrium of $\Gamma$.
\end{definition}

In an $\varepsilon$-normal-form perfect equilibrium, all pure strategies that are not best responses are subject to the same probability upper bound. \cite{MyersonRefinementsNashequilibrium1978} refined this notion by requiring, in an $\varepsilon$-normal-form proper equilibrium, that each lower-payoff pure strategy receive probability at most $\varepsilon$ times that of any higher-payoff pure strategy.

\begin{definition}\label{nfpr-def-pro1}
	Let $\Gamma$ be an extensive form game. For any sufficiently small $\varepsilon > 0$, a totally mixed strategy profile $\sigma(\varepsilon)\in \Xi$ is an $\varepsilon$-normal-form proper equilibrium of $\Gamma$ if $\sigma^i(\varepsilon;s^i)\leq\varepsilon\sigma^i(\varepsilon;\tilde s^i)$ whenever $u^i(s^i, \sigma^{-i}(\varepsilon)) < u^i(\tilde s^{i},\sigma^{-i}(\varepsilon))$ for all $i\in N$ and $s^i,\tilde s^i\in S^i$. A mixed strategy profile $\sigma^*\in \Xi$ is defined as a normal-form proper equilibrium of game $\Gamma$ if $\sigma^*$ is a limit point of some sequence $\{\sigma(\varepsilon^k)\}_{k=1}^\infty$, where $\lim_{k\to\infty}\varepsilon^k=0$ and each $\sigma(\varepsilon^k)$ is an $\varepsilon^k$-normal form proper equilibrium of $\Gamma$.
\end{definition}

\begin{table}[tbp]
	\centering
	\caption{Notation for Games in Normal Form or Sequence Form}
	\label{nfpr-tab-pre2}
	\fontsize{10pt}{11pt}\selectfont
	\setlength{\tabcolsep}{3pt}
	\begin{tabular}{ll}
		\toprule
		Symbol & Explanation\\
		\midrule
		$s^i$ & Pure strategy of player $i$\\
		$S=\prod_{i\in N_c}S^i$ & Set of pure strategy profiles\\
		$\sigma^i$ & Mixed strategy of player $i\in N_c$, probability measure over $S^i$\\
		$\Xi={\prod_{i\in N}}\Xi^i$ & Set of mixed strategy profiles,  $\Xi^i=\{\sigma^i:S^i\to\mathbb{R}_+\mid \sum\limits_{s^i\in S^i}\sigma^{i}(s^i)=1\}$\\
		$\Xi_{++}={\prod_{i\in N}} \Xi^i_{++}$ & Set of strictly positive mixed strategy profiles\\
		$u^i(s)$ & Expected payoff of player $i$ on the pure strategy profile $s$\\
		$\varpi^i$ & Sequence of actions taken by player $i$\\
		$\varpi^i_{I^j_i}$ & Sequence of player $i$ leading to $I^j_i$, $\varpi^i_h=\varpi^i_{I^j_i}$ for any $h\in I^j_i$ \\
		$\varpi^i_{I^j_i}a$ & The extended sequence $\varpi^i_{I^j_i}\cup \{a\}$\\
		${W}=\prod_{i\in N_c}{W}^i$ & The collection of sequence profiles, $\varpi^i_\emptyset\in{W}^i$\\
		$g^i(\varpi)$ & Expected payoff of player $i$ on the sequence profile $\varpi$\\
		$\gamma^i$ & Realization plan of player $i\in N_c$\\
		$\Lambda=\prod_{i\in N}{ \Lambda^i}$ & Set of realization plan profiles\\
		$\Lambda_{++}=\prod_{i\in N}{ \Lambda^i_{++}}$ & Set of  strictly positive realization plan profiles\\
		$M_i(\varpi^i)$ & Indices of player \(i\)'s information sets with $\varpi^i$ being the sequence\\
		\bottomrule
	\end{tabular}
\end{table}

The sequence form of $\Gamma$ is represented as $\Gamma_s=\langle N,\{{W}^i\}_{i\in N_c},\gamma^c,\{g^i\}_{i\in N}\rangle$. The notation used here is summarized in Table~\ref{nfpr-tab-pre2}. For $i\in N_c$, a sequence $\varpi^i$ records player $i$'s actions along a history. Specifically, for $h=\langle a_1,\ldots,a_L\rangle\in H$, let $\varpi^i_h=\{a_q\mid\text{$a_q\in A(I^j_i)$ for some $j\in M_i$ and $1\leq q\leq L$}\}$, which is either the empty sequence $\varpi^i_\emptyset$ or an extension $\varpi^i_{I^j_i}a$ of a preceding sequence $\varpi^i_{I^j_i}$ with $j \in M_i$ and $a \in A(I^j_i)$. A sequence profile $\varpi=(\varpi^i\in N_c)\in W$ is said to be defined by the history $h\in H$ if $\varpi^i=\varpi_h^i$ for all $i\in N_c$; the payoff function $g^i$ is then defined by setting $g^i(\varpi)=u_z^i(h)$ when $\varpi$ is defined by some terminal history $h\in Z$, and \(g^i(\varpi)=0\) otherwise. For each player $i\in N_c$, a realization plan is a function $\gamma^i$ defined on ${W}^i$ that satisfies $\gamma^i(\varpi^i_\emptyset)=1$ and
\begin{equation}\label{nfpr-equ-pre3}
	\begin{array}{l}
		\sum\limits_{a\in A(I_i^j)}\gamma^i(\varpi_{I_i^j}^i a)-\gamma^i(\varpi_{I_i^j}^i)=0,\; j\in M_i,\\
		0\le \gamma^i(\varpi_{I_i^j}^i a),\; j\in M_i,a\in A(I_i^j).
	\end{array}
\end{equation}
The chance player's realization plan $\gamma^c$ is given by $\gamma^c(\varpi^c)=\prod_{a\in \varpi^c\cap A(h)}f_c(a|h)$ with $\varpi^c\in W^c$, which satisfies System~(\ref{nfpr-equ-pre3}). Given a realization plan profile $\gamma=(\gamma^i:i\in N_c)$, the expected payoff of player $i\in N$ is $g^i(\gamma)=\sum_{\varpi^i\in W^i}\gamma^i(\varpi^i)g^i(\varpi^i,\gamma^{-i})$, where \[\textstyle g^i(\varpi^i,\gamma^{-i})=\sum\limits_{\varpi^{-i}\in {W}^{-i}}g^i(\varpi^i,\varpi^{-i})\prod\limits_{i_q\ne i}\gamma^{i_q}(\varpi^{i_q}).\]For player $i$, the number of sequences is $\sum_{j \in M_i} |A(I_i^j)|+1$, growing linearly with the number of information sets. This compactness underlies the development of efficient methods for computing normal-form proper equilibria in extensive-form games.

\begin{figure}[htbp]
	\centering
	
	\begin{minipage}[t]{0.4\textwidth}
		\vspace{0pt}
		\centering
		\includegraphics[width=\textwidth]{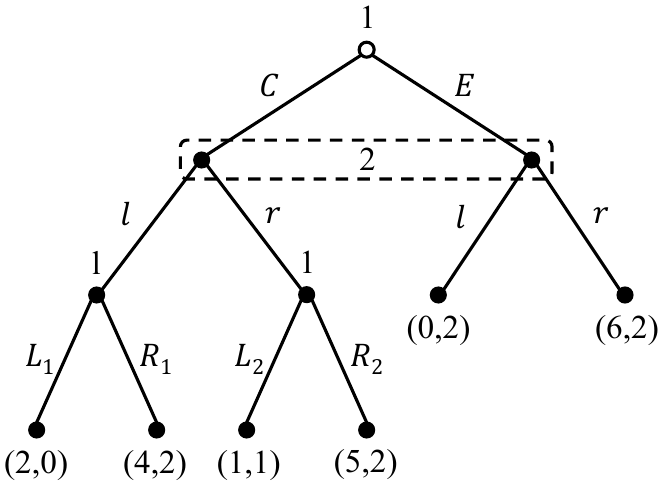}
		\captionof{figure}{An extensive-form game}
		\label{fig:game1}
	\end{minipage}
	\hfill
	\begin{minipage}[t]{0.58\textwidth}
		\vspace{0pt}
		\centering
		\captionof{table}{Reduced Normal Form of Figure~\ref{fig:game1}}
		\label{tab:normalform}
		\fontsize{10pt}{11pt}\selectfont
		\setlength{\tabcolsep}{3pt}
		\begin{tabular}{lcc}
			\toprule
			& \multicolumn{2}{c}{Player 2 pure strategies} \\
			\cmidrule(r){2-3}
			Player 1 pure strategies
			& \(s_1^2=\{l\}\)
			& \(s_2^2=\{r\}\) \\
			\midrule 
			\(s_1^1=\{C,L_1,L_2\}\) & \((2,0)\) & \((1,1)\) \\
			\(s_2^1=\{C,L_1,R_2\}\) & \((2,0)\) & \((5,2)\) \\
			\(s_3^1=\{C,R_1,L_2\}\) & \((4,2)\) & \((1,1)\) \\
			\(s_4^1=\{C,R_1,R_2\}\) & \((4,2)\) & \((5,2)\) \\
			\(s_5^1=\{E\}\)         & \((0,2)\) & \((6,2)\) \\
			\bottomrule
		\end{tabular}
	\end{minipage}
	
\end{figure}

\begin{example}\label{exm1}{
		Consider the two-player extensive-form game $\Gamma$ shown in Figure~\ref{fig:game1}. The players' information sets are given by \( I^1_1=\{\emptyset\} \), \( I^2_1=\{\langle C,l\rangle\} \), \( I^3_1=\{\langle C,r\rangle\} \), and \( I^1_2=\{\langle C\rangle,\langle E\rangle\} \). The corresponding action sets are \( A(I^1_1)=\{C,E\} \), \( A(I^2_1)=\{L_1,R_1\} \), \( A(I^3_1)=\{L_2,R_2\} \), and \( A(I^1_2)=\{l,r\} \). Table~\ref{tab:normalform} presents the reduced normal form, with pure strategies represented by their prescribed action set and entries $(u_1,u_2)$ giving the resulting payoff pairs. The mixed strategies can be expressed in vector form as
		$\bm{\sigma^1}=(\sigma^1(s^1_1),\sigma^1(s^1_2),\sigma^1(s^1_3),\sigma^1(s^1_4),\sigma^1(s^1_5))^\top$, $\bm{\sigma^2}=(\sigma^2(s^2_1),\sigma^2(s^2_2))^\top$. By Definition~\ref{nfpr-def-pro1}, the normal-form proper equilibria of this game and their corresponding expected payoffs $u(\sigma)=(u^1(\sigma),u^2(\sigma))$ can be derived analytically as follows:
		\begin{enumerate}
			\item $\bm{\sigma^1} = (0,0,0,0,1)^\top$, $\bm{\sigma^2} = (0, 1)^\top$; payoff $u(\sigma)=(6,2)$.
			\item $\bm{\sigma^1} = (0,0,0,1,0)^\top$, $\bm{\sigma^2} = (1, 0)^\top$; payoff $u(\sigma)=(4,2)$.
			\item $\bm{\sigma^1} = (0,0,0,1,0)^\top$, $\bm{\sigma^2} = (\frac{2}{3}, \frac{1}{3})^\top$; payoff $u(\sigma)=(\frac{13}{5},2)$.
		\end{enumerate}
		The sequence-form representation is summarized in Table~\ref{tab:sequenceform} and the realization plans satisfy
		\[
		\begin{array}{l}
			\gamma^1(\varpi^1_{I^1_1}C)+\gamma^1(\varpi^1_{I^1_1}E)=\gamma^1(\varpi^1_\emptyset)=1,\; 0\le\gamma^1(\varpi^1_{I^1_1}E),\\
			\gamma^1(\varpi^1_{I^2_1}L_1)+\gamma^1(\varpi^1_{I^2_1}R_1)=\gamma^1(\varpi^1_{I^1_1}C),\; 0\le\gamma^1(\varpi^1_{I^2_1}L_1),\; 0\le\gamma^1(\varpi^1_{I^2_1}R_1),\\
			\gamma^1(\varpi^1_{I^3_1}L_2)+\gamma^1(\varpi^1_{I^3_1}R_2)=\gamma^1(\varpi^1_{I^1_1}C),\; 0\le\gamma^1(\varpi^1_{I^3_1}L_2),\; 0\le\gamma^1(\varpi^1_{I^3_1}R_2),\\
			\gamma^2(\varpi^2_{I^1_2}l)+\gamma^2(\varpi^2_{I^1_2}r)=\gamma^2(\varpi^2_\emptyset)=1,\; 0\le\gamma^2(\varpi^2_{I^1_2}l),\; 0\le\gamma^2(\varpi^2_{I^1_2}r).
		\end{array}
		\]
	}
	\begin{table}[htbp]
		\centering
		\caption{Sequence-Form Representation of Figure~\ref{fig:game1}}
		\label{tab:sequenceform}
		\fontsize{10pt}{11pt}\selectfont
		\setlength{\tabcolsep}{3pt}
		\begin{tabular}{cccccccc}
			\toprule
			& \multicolumn{7}{c}{Player 1 sequences} \\
			\cmidrule(r){2-8}
			Player 2 sequences
			& \(\varpi^1_\emptyset\)
			& \(\varpi^1_{I^1_1}C\)
			& \(\varpi^1_{I^1_1}E\)
			& \(\varpi^1_{I^2_1}L_1\)
			& \(\varpi^1_{I^2_1}R_1\)
			& \(\varpi^1_{I^3_1}L_2\)
			& \(\varpi^1_{I^3_1}R_2\) \\
			\midrule
			
			\(\varpi^2_\emptyset\)
			& \((0,0)\) & \((0,0)\) & \((0,0)\)
			& \((0,0)\) & \((0,0)\) & \((0,0)\) & \((0,0)\) \\
			
			\(\varpi^2_{I^1_2}l\)
			& \((0,0)\) & \((0,0)\) & \((0,2)\)
			& \((2,0)\) & \((4,2)\) & \((0,0)\) & \((0,0)\) \\
			
			\(\varpi^2_{I^1_2}r\)
			& \((0,0)\) & \((0,0)\) & \((6,2)\)
			& \((0,0)\) & \((0,0)\) & \((1,1)\) & \((5,2)\) \\
			\bottomrule
		\end{tabular}
	\end{table}
\end{example}
\section{A Sequence-Form Characterization of Normal-Form Proper Equilibria}\label{nfpr-sec-prm2}
Consider an extensive-form game, $\Gamma$, with $\Gamma_n$ as its normal form and $\Gamma_s$ as its sequence form. Given any $s^i\in S^i$, $i\in N_c$, we define $s^i(\varpi^i)=\prod_{a\in \varpi^i}s^i(a)$ for $\varpi^i\in W^i$. Thus, $s^i(\varpi^i)$ equals $0$ if  $s^i$ prescribes all actions in $\varpi^i$, and $0$ otherwise. For any $\sigma\in \Xi$, let $\gamma(\sigma)=(\gamma^{i}(\sigma^i;\varpi^i):i\in N_c,\varpi^i\in W^i)$, where 
\begin{equation}\label{ms2rp}\textstyle
	\gamma^{i}(\sigma^i;\varpi^i)=\sum\limits_{s^i\in S^i}s^i(\varpi^i)\sigma^i(s^i),\;i\in N_c,\varpi^i\in W^i.
\end{equation}
Then, $\gamma^{i}(s^i;\varpi^i)=s^i(\varpi^i)$ and $\gamma(\sigma)\in\Lambda$. Let $T=\left\{(\sigma,\gamma)\left|\sigma\in  \Xi,\gamma=\gamma(\sigma)\right.\right\}$, which leads to the following conclusions. For any $\gamma\in\Lambda$, there exists a mixed strategy profile $\sigma$ such that $(\sigma,\gamma)\in T$. In particular, for any \(\gamma \in \Lambda_{++}\), one such profile is $\sigma(\gamma)=(\sigma^i(\gamma^i): i\in N_c)$, where $\sigma^i(\gamma^i;s^i)	= \prod_{j\in M_i,\, a\in A(I^j_i),\,s^i(a)=1}{\gamma^i(\varpi^i_{I^j_i}a)}/{\gamma^i(\varpi^i_{I^j_i})}$. Moreover, if $(\sigma,\gamma)\in T$, then $u^i(\sigma)=g^i(\gamma)$ for all $i\in N$. Rigorous proofs of these results can be found in \cite{houSequenceFormCharacterizationDifferentiable2025}.

\subsection{Free Strategy Variables}
For any $I_i^j,i\in N,j\in M_i$, once $\gamma^i(\varpi^i_{I_i^j})$ is fixed, the constraints~\eqref{nfpr-equ-pre3} implies that only $|A(I_i^j)|-1$ variables among $\{\gamma^i(\varpi^i_{I_i^j}a)\mid a\in A(I_i^j)\}$ are free. For any sequence $\varpi^i\in W^i$, let $M_i^{\mathrm{d}}(\varpi^i)=\{j\in M_i\mid \varpi^i\subseteq\varpi^i_{I^j_i}\}$, and define $\kappa_i(\varpi^i)=1+\sum_{j\in M^{\mathrm{d}}_i(\varpi^i)}(|A(I^j_i)|-1)$. Thus, $\kappa_i(\varpi^i)$ counts the number of free realization-plan variables associated with the sequences $\{\tilde\varpi^i\in W^i\mid \varpi^i\subseteq\tilde\varpi^i\}$. In particular, $\kappa_i(\varpi^i_\emptyset)-1$ gives the intrinsic dimension of player $i$'s realization-plan space, since $\gamma^i(\varpi^i_\emptyset)=1$ is fixed. Let $K^-_i=\{1,2,\ldots,\kappa_i(\varpi^i_\emptyset)-1\}$ and $K_i=K^-_i\cup\{\kappa_i(\varpi^i_\emptyset)\}$.

\begin{definition}\label{seqorder}
	We define $\pi^i_-=\{\varpi^i_1,\varpi^i_2,\dots,\varpi^i_{\kappa_i(\varpi^i_\emptyset)-1}\}$ as an ordering of the sequences associated with player $i$'s $\kappa_i(\varpi^i_\emptyset)-1$ free realization-plan variables such that
	\begin{enumerate}
		\item for any $j\in M_i$, the intersection $\{\varpi^i_{I^j_i}a \mid a\in A(I^j_i)\}\cap\pi^i_-$ has cardinality $|A(I^j_i)|-1$,
		\item for every $l\in K^-_i$, no element $\varpi^i_{l'}$ with $l'>l$ is an extended sequence of $\varpi^i_{l}$.
	\end{enumerate}
\end{definition}

Specifying the realization probabilities for the sequences in $\pi^i_-$ uniquely determines all remaining realization-plan variables via (\ref{nfpr-equ-pre3}). We further set $\varpi^i_{\kappa_i(\varpi^i_\emptyset)}$ to be the empty sequence and extend the ordering as $\pi^i=\pi^i_- \cup \{\varpi^i_{\kappa_i(\varpi^i_\emptyset)}\}$. Let $\Pi^i$ denote the collection of all such extended orderings. Given any $\gamma=(\gamma^i:i\in N)\in\Lambda_{++}$, we construct, for each player $i$, an ordering $\pi^i(\gamma^i) \in \Pi^i$ as follows. At each information set $I^j_i$, $j\in M_i$, choose one sequence $\varpi^i_{I^j_i}a^\ast_j$ such that $a_j^\ast\in\arg\max_{a\in A(I_i^j)}\gamma^i(\varpi^i_{I_i^j}a)$. The ordering $\pi^i(\gamma^i)$ is then obtained by arranging all sequences in $W^i\setminus\{\varpi^i_{I_i^j}a_j^\ast:j\in M_i\}$ in ascending order of their realization probabilities.

\begin{lemma}\label{freevar}
	Given $\pi^i\in \Pi^i$ and $\gamma^i\in \Lambda^i$, a mixed strategy $\sigma^i$ satisfies $\gamma^i(\sigma^i;\varpi^i)=\gamma^i(\varpi^i)$ for all $\varpi^i\in W^i$ if and only if this equality holds for all $\varpi^i\in \pi^i$. Such a mixed strategy has exactly $|S^i|-\kappa_i(\varpi^i_\emptyset)$ free variables.
\end{lemma}

For subsequent use, we specify a particular choice of free mixed-strategy variables in Lemma~\ref{freevar}. For each $\varpi_l^i\in\pi^i$, $l\in K_i$, select $s_l^i\in S^i$ such that $s_l^i(\varpi_l^i)=1$ and, whenever $s_l^i(\varpi^i)=1$ for some $\varpi^i\not\subseteq\varpi_l^i$, we have $\varpi^i\in W^i\setminus\pi^i$. Define $S^i(\pi^i)=\{s_1^i,s_2^i,\dots,s_{\kappa_i(\varpi^i_\emptyset)}^i\}$ and the mixed-strategy free variables are then given by $\{\sigma^i(s^i)\mid s^i\in S^i\setminus S^i(\pi^i)\}$.

\begin{example}
	Consider player~1 of the game in Figure~\ref{fig:game1}. The realization plan contains three free variables and  $\kappa_1(\varpi^1_\emptyset)=4$. One sequence ordering in $\Pi^1$ is $\pi^1=\{\varpi^1_{I^2_1}L_1,\varpi^1_{I^3_1}L_2,\varpi^1_{I^1_1}E,\varpi^1_\emptyset\}$. Fix $\gamma^1(\varpi^1_{I^2_1}L_1)=\lambda_1$, $\gamma^1(\varpi^1_{I^3_1}L_2)=\lambda_2$, $\gamma^1(\varpi^1_{I^1_1}E)=\lambda_3$. It follows from \eqref{nfpr-equ-pre3} that $\gamma^1(\varpi^1_{I^1_1}C)=1-\lambda_3$, $\gamma^1(\varpi^1_{I^3_1}R_2)=1-\lambda_3-\lambda_2$, and $\gamma^1(\varpi^1_{I^2_1}R_1)=1-\lambda_3-\lambda_1$.
	Specifically, let $\lambda_1=0.1$, $\lambda_2=0.2$, $\lambda_3=0.3$, the induced ordering $\pi^1(\gamma^1)$ is consistent with $\pi^1$. The mixed strategy $\sigma^1$ satisfying $\gamma^1(\sigma^1;\varpi^1)=\gamma^1(\varpi^1)$ for all $\varpi^1\in W^1$ has one free variable. Moreover, $S^1(\pi^1)=\{s_2^1,s_3^1,s_5^1,s_4^1\}$, and therefore $S^1\setminus S^1(\pi^1)=\{s_1^1\}$. If $\sigma^1(s_1^1)=\lambda_4$, then \eqref{ms2rp} yields the remaining components as
	$\sigma^1(s_2^1)=\gamma^1(\varpi^1_{I^2_1}L_1)-\lambda_4$, $\sigma^1(s_3^1)=\gamma^1(\varpi^1_{I^3_1}L_2)-\lambda_4$, $\sigma^1(s_5^1)=\gamma^1(\varpi^1_{I^1_1}E)$, and $\sigma^1(s_4^1)=\gamma^1(\varpi^1_{I^2_1}R_1)-\gamma^1(\varpi^1_{I^3_1}L_2)+\lambda_4$.
\end{example}

\subsection{Sequence-Form Proper Equilibrium}
This subsection introduces sequence-form proper equilibrium and, building on the preceding subsection, establishes its strategic equivalence to normal-form proper equilibrium. Given \(\gamma\in\Lambda\), for each player $i\in N$ and $\varpi^i\in W^i$, we define
\begin{equation}\label{maxpayoff}
	\textstyle g^i_m(\varpi^i,\gamma^{-i})=\max\limits_{\tilde\gamma^i\in \Lambda^i,\tilde\gamma^i(\varpi^i)=1} g^i(\tilde\gamma^i,\gamma^{-i}).
\end{equation}
Unlike \(g^i(\varpi^i,\gamma^{-i})\), which captures only the local payoff contribution of \(\varpi^i\) to \(g^i(\gamma)\), \(g^i_m(\varpi^i,\gamma^{-i})\) evaluates the maximal expected payoff attainable by player \(i\in N\) conditional on realizing \(\varpi^i\), thereby allowing payoff comparisons across all sequences. To illustrate, Table~\ref{tab:gmexample} contrasts $g^1(\varpi^1,\gamma^2)$ with $g^1_m(\varpi^1,\gamma^2)$ for all $\varpi^1\in W^1$ in the game of Figure~\ref{fig:game1}, fixing $\gamma^2(\varpi^2_{I^1_2}l)=\gamma^2(\varpi^2_{I^1_2}r)=1/2$. 
\begin{table}[htbp]
	\centering
	\caption{Comparison of \(g^1(\varpi^1,\gamma^2)\) and
		\(g^1_m(\varpi^1,\gamma^2)\)}
	\label{tab:gmexample}
	\fontsize{10pt}{11pt}\selectfont
	\setlength{\tabcolsep}{3pt}
	\begin{tabular}{lccccccc}
		\toprule
		\(\varpi^1\)
		& \(\varpi^1_\emptyset\)
		& \(\varpi^1_{I^1_1}C\)
		& \(\varpi^1_{I^1_1}E\)
		& \(\varpi^1_{I^2_1}L_1\)
		& \(\varpi^1_{I^2_1}R_1\)
		& \(\varpi^1_{I^3_1}L_2\)
		& \(\varpi^1_{I^3_1}R_2\) \\
		\midrule
		
		\(g^1(\varpi^1,\gamma^2)\)
		& \(0\) & \(0\) & \(3\) & \(1\) & \(2\) & \(1/2\) & \(5/2\) \\
		
		\(g^1_m(\varpi^1,\gamma^2)\)
		& \(9/2\) & \(9/2\) & \(3\) & \(7/2\) & \(9/2\) & \(5/2\) & \(9/2\) \\
		\bottomrule
	\end{tabular}
\end{table}
Using the maximal sequence payoff in \eqref{maxpayoff}, we define sequence-form proper equilibrium as follows.

\begin{definition}\label{nfpr-def-sfc1}
	Let $\Gamma$ be an extensive form game. For any sufficiently small $\varepsilon > 0$, $\gamma(\varepsilon)\in \Lambda_{++}$ is an $\varepsilon$-sequence-form proper equilibrium of $\Gamma$ if $\gamma^i(\varepsilon;\varpi^i)\leq\varepsilon\gamma^i(\varepsilon;\tilde \varpi^i)$ whenever $g^i_m(\varpi^i,\gamma^{-i}(\varepsilon))\! < \! g^i_m(\tilde \varpi^{i},\gamma^{-i}(\varepsilon))$ for all $i\in N$ and $\varpi^i,\tilde \varpi^i \!\in\! W^i$. A realization plan profile $\gamma^*\in \Lambda$ is defined as a sequence-form proper equilibrium of game $\Gamma$ if $\gamma^*$ is a limit point of some sequence $\{\gamma(\varepsilon^k)\}_{k=1}^\infty$, where $\lim_{k\to\infty}\varepsilon^k=0$ and each $\gamma(\varepsilon^k)$ is an $\varepsilon^k$-sequence-form proper equilibrium of $\Gamma$.
\end{definition}

We next establish the equivalence between normal-form proper equilibria and sequence-form proper equilibria. To this end, we first introduce a best-response notion for sequences.
\begin{definition}\label{bestresponse}
	Consider a realization plan profile $\gamma\in \Lambda$. For any $i\in N,j\in M_i,a\in A(I^j_i)$, we refer to $\varpi^i_{I^j_i}a$ as an $I^j_i$-best-response sequence to $\gamma$ if the following equality holds for any $a'\in A(I^j_i)$,
	\begin{equation}\label{defbest}\textstyle g^i_m(\varpi^i_{I^j_i}a,\gamma^{-i})\geq g^i_m(\varpi^i_{I^j_i}a',\gamma^{-i}).\end{equation}We define $\varpi^i_{I^j_i}a$ as a best-response sequence to $\gamma$ if, for any $j_q\in M_i,a_q\in A(I^{j_q}_i)$ with $a_q\in \varpi^i_{I^j_i}a$, $\varpi^i_{I^{j_q}_i}a_q$ qualifies as an $I^{j_q}_i$-best-response sequence to $\gamma$.
\end{definition}

The following lemma links optimal pure strategies to best-response sequences.
\begin{lemma}\label{nfprpure}{\em
		Let $(\sigma,\gamma)\in T$, and consider $\varpi^i\in W^i,s^i\in S^i$ of player $i$ with $s^i(\varpi^i)=1$. The following statements are equivalent:
		\begin{enumerate}
			\item For all $\tilde s^i\in S^i$ satisfying $\tilde s^i(\varpi^i)=1$, it holds that $u^i(s^i,\sigma^{-i})\geq u^i(\tilde s^i,\sigma^{-i})$.
			\item For any sequence $\varpi^i_{I^j_i}a\nsubseteq\varpi^i$, with $j\in M_i,a\in A(I^j_i)$, such that $s^i(\varpi^i_{I^j_i}a)=1$, $\varpi^i_{I^j_i}a$ is a $I^j_i$-best-response sequence to $\gamma$.
	\end{enumerate} }
\end{lemma}

The proof of Lemma~\ref{nfprpure} is given in Appendix~\ref{prooflem2}. Thus, the first statement of Lemma~\ref{nfprpure} implies that $g^i_m(\varpi^i,\gamma^{-i}) = u^i(s^i,\sigma^{-i})$, which brings us the following equivalence result.
\begin{theorem}\label{nfpr-thm-sfc3}
	For $(\sigma^*,\gamma^*)\in T$, $\sigma^*$ is a normal-form proper equilibrium if and only if $\gamma^*$ is a sequence-form proper equilibrium.
\end{theorem}
\begin{proof}
$(\Rightarrow)$ Suppose that $\gamma^*$ is the limit point of a convergent sequence $\{\gamma(\varepsilon^k)\}_{k=1}^\infty$ with $\lim_{k\to\infty} \varepsilon^k = 0$, where $\gamma(\varepsilon^k)$ denotes an $\varepsilon^k$-sequence-form proper equilibrium. For each $k$, multiple totally mixed strategies $\sigma^k$ satisfy $(\sigma^k,\gamma(\varepsilon^k))\in T$. We establish that there exists at least one such $\sigma^k$ that constitutes an $\tilde{\varepsilon}^k$-normal-form proper equilibrium for some $\tilde{\varepsilon}^k>0$.

Let $w_0=\max_{i\in N}\{|S^i|-\kappa_i(\varpi^i_\emptyset)+1\}$ and define $\tilde\varepsilon_s^k={(\varepsilon^k)}^{\frac{1}{w_0}}$. For each player $i\in N$, we construct $\sigma^{ki}$ as follows. Let $\pi^i(\gamma^i(\varepsilon^k))= \{\varpi^i_1,\varpi^i_2,\dots,\varpi^i_{\kappa_i(\varpi^i_\emptyset)}\}$ and $S^i(\pi^i(\gamma^i(\varepsilon^k)))= \{s^i_1,s^i_2,\dots,s^i_{\kappa_i(\varpi^i_\emptyset)}\}$. For any $i\in N,j\in M_i,a\in A(I^j_i)$ such that $\varpi^i_{I^j_i}a\in W^i\setminus\pi^i(\gamma^i(\varepsilon^k))$, $\varpi^i_{I^j_i}a$ is an $I^j_i$-best-response sequence to $\gamma(\varepsilon^k)$, because $\gamma^i(\varepsilon^k;\varpi^i_{I^j_i}a)\geq \gamma^i(\varepsilon^k;\varpi^i_{I^j_i}a')$ derives $g^i_m(\varpi^i_{I^j_i}a,\gamma^{-i}(\varepsilon^k))\geq g^i_m(\varpi^i_{I^j_i}a',\gamma^{-i}(\varepsilon^k))$ for all $a'\in A(I^j_i)$. Moreover, by the construction of $s^i_l$ for $l\in K_i$ and Lemma~\ref{nfprpure}, one have $g^i_m(\varpi^i_l,\gamma^{-i}(\varepsilon^k))=u^i(s^i_l,\sigma^{-i}(\gamma(\varepsilon^k)))$ for $l\in K_i$. Order the pure strategies in ascending expected payoff $u^i(s^i,\sigma^{-i}(\gamma(\varepsilon^k)))$, placing $s^i_{\kappa_i(\varpi^i_\emptyset)}$ at the last position, yielding $\{\tilde s^i_q\}_{q=1}^{|S^i|}$. We set
\[\textstyle
\sigma^{ki}(\tilde s^i_q)=\begin{cases}
	\gamma^i(\varepsilon^k;\varpi^i_l)&\text{if $\tilde s^i_q=s^i_l$ for some $l\in K_i$,}\\
	\max\{\sigma^{*i}(\tilde s^i_q), \; \tilde\varepsilon_s^k \sigma^{ki}(\tilde s^i_{q+1})\}&\text{otherwise.}
\end{cases}
\]
Subsequent adjustment of $\sigma^{ki}(s^i_l)$ for $l\in K_i$ according to (\ref{ms2rp}) completes the construction. When $k$ is sufficiently large, take any $1\leq q\leq |S^i|-1$ with $u^i(\tilde s^i_q,\sigma^{-i}(\gamma(\varepsilon^k)))<u^i(\tilde s^i_{q+1},\sigma^{-i}(\gamma(\varepsilon^k)))$.
\begin{itemize}
	\item If $\tilde s^i_{q+1}=s^i_l$ holds for some $l\in K_i$, then we have $\sigma^{ki}(\tilde s^i_{q+1})\geq \gamma^i(\varepsilon^k;\varpi^i_l)-\gamma^i(\varepsilon^k;\varpi^i_l)\sum^{|S^i|-\kappa_i(\varpi^i_\emptyset)}_{w=1}(\tilde\varepsilon_s^k)^w\geq \sqrt{\tilde\varepsilon_s^k}\gamma^i(\varepsilon^k;\varpi^i_l)\geq 1/\sqrt{\tilde\varepsilon_s^k}\sigma^{ki}(\tilde s^i_{q})$.
	\item If $\tilde s^i_{q}=s^i_l$ holds for some $l\in K_i$, then we have $\sigma^{ki}(\tilde s^i_{q})\leq \gamma^i(\varepsilon^k;\varpi^i_l)\leq (\tilde\varepsilon_s^k)^{|S^i|-\kappa_i(\varpi^i_\emptyset)+1}\gamma^i(\varepsilon^k;\varpi^i_{l+1})\leq{\tilde\varepsilon_s^k}\sigma^{ki}(\tilde s^i_{q+1})\leq \sqrt{\tilde\varepsilon_s^k}\sigma^{ki}(\tilde s^i_{q+1})$.
	\item Otherwise, we have $\sigma^{ki}(\tilde s^i_q)=\tilde\varepsilon_s^k \sigma^{ki}(\tilde s^i_{q+1})\leq \sqrt{\tilde\varepsilon_s^k}\sigma^{ki}(\tilde s^i_{q+1})$.
\end{itemize} 
Thus, $\sigma^k$ is an $\sqrt{\tilde\varepsilon_s^k}$-normal-form proper equilibrium with $(\sigma^k,\gamma(\varepsilon^k))\in T$ and $\lim_{k\to\infty} \sigma^k = \sigma^*$.

$(\Leftarrow)$ Conversely, assume that $\sigma^*$ is a normal-form proper equilibrium of $\Gamma$, and let $\gamma^*=\gamma(\sigma^*)$. Then, there exists a sequence of totally mixed strategies $\{\sigma(\varepsilon^k)\}_{k=1}^{\infty}$ such that $\lim_{k\to\infty}\varepsilon^k=0$ and $\lim_{k\to\infty}\sigma(\varepsilon^k)=\sigma^*$, where each $\sigma(\varepsilon^k)$ is an $\varepsilon^k$-normal form proper equilibrium. Consider a specific $\sigma(\varepsilon^k)$ with sufficiently large $k$. If $u^i(s^i, \sigma^{-i}(\varepsilon^k)) < u^i(\tilde s^{i},\sigma^{-i}(\varepsilon^k))$ holds for some $\tilde s^i\in S^i$, then $\sigma^i(\varepsilon^k;s^i)\leq\varepsilon^k\sigma^i(\varepsilon^k;\tilde s^i)$. We prove that $\gamma(\sigma(\varepsilon^k))$ is a $\sqrt{\varepsilon^k}$-sequence-form proper equilibrium for all $k$ with sufficiently small $\varepsilon^k$. 

Consider any two sequences $\varpi^i,\tilde\varpi^i\in W^i$ with $g^i_{m}(\varpi^i,\gamma^{-i}(\sigma(\varepsilon^k)))< g^i_{m}(\tilde\varpi^i,\gamma^{-i}(\sigma(\varepsilon^k)))$. By Equation (\ref{maxpayoff}) and Lemma~\ref{nfprpure}, one can find $s^{*i}, \tilde s^{*i}\in S^i$ satisfying $s^{*i}(\varpi^i)=1$, $\tilde s^{*i}(\tilde\varpi^i)=1$,  $u^i(s^{*i},\sigma^{-i}(\varepsilon^k))=g^i_{m}(\varpi^i,\gamma^{-i}(\sigma(\varepsilon^k)))$, and $u^i(\tilde s^{*i},\sigma^{-i}(\varepsilon^k))=g^i_{m}(\tilde\varpi^i,\gamma^{-i}(\sigma(\varepsilon^k)))$. Moreover, we obtain $u^i(s^{*i},\sigma^{-i}(\varepsilon^k))=\max_{s^i\in S^i,s^i(\varpi^i)=1}u^i(s^{i},\sigma^{-i}(\varepsilon^k))$. It follows that, for any $s^{i}\in S^i$ with $s^{i}(\varpi^i)=1$, $\sigma^i(\varepsilon^k;s^{i})\leq\varepsilon^k\sigma^i(\varepsilon^k;\tilde s^{*i})$. Then, $\gamma^{i}(\sigma(\varepsilon^k);\varpi^i)=\sum_{s^i\in S^i}s^i(\varpi^i)\sigma^i(\varepsilon^k;s^i)\leq\sqrt{\varepsilon^k}\sigma^i(\varepsilon^k;\tilde s^{*i})\leq\sqrt{\varepsilon^k}\sum_{s^i\in S^i}s^i(\tilde\varpi^i)\sigma^i(\varepsilon^k;s^i)=\sqrt{\varepsilon^k}\gamma^{i}(\sigma(\varepsilon^k);\tilde\varpi^i)$, which completes the proof.
\end{proof}
Following the approach developed in this subsection, we discuss in Appendix~\ref{formulations} equivalent sequence-form characterizations of Nash equilibrium and several other refinements.

\section{An Equivalent Formulation of Sequence-Form Proper Equilibrium by Perturbed Games}\label{nfpr-sec-prm3}
This subsection presents a class of perturbed sequence-form games and establishes the equivalence between the limits of their Nash equilibria and sequence-form proper equilibria, yielding a characterization suitable for computation. Let $\delta(\varepsilon) = (\delta^i(\varepsilon) : i \in N)$ with 
$\delta^i(\varepsilon) = (\delta^i_l(\varepsilon) : l \in K^-_i)$ 
and $0 < \delta^i_l(\varepsilon) \leq \varepsilon$. For each \(i\in N\), define $r_{l}^{i}(\delta(\varepsilon))=\prod^{\kappa_i(\varpi^i_\emptyset)-l}_{q=1}\delta^i_q(\varepsilon)$, \(l\in K^-_i\). To obtain the intended perturbation hierarchy, we impose probability lower bounds only on restricted sequence sets.

\begin{definition}\label{seqcons}
	For $i\in N$, define $\mathcal{E}^i$ as the collection of all sequence sets $E^i\subset W^i\setminus\{\varpi^i_\emptyset\}$ that satisfies the following two conditions,
	\begin{enumerate}
		\item For any $j\in M_i$, there exists $a\in A(I^j_i)$ such that $\varpi^i_{I^j_i}a\notin E^i$.
		\item If $\varpi^i,\tilde \varpi^i\in E^i$, then $\tilde \varpi^i\nsubseteq \varpi^i$.
	\end{enumerate}
\end{definition}

Let $e_0=\sum_{i\in N}|\mathcal{E}^i|$. For $i\in N,E^i\in \mathcal{E}^i$, set $q^i(E^i) = \sum_{\varpi^i\in E^i}\kappa_i(\varpi^i)$. Then, for sufficiently small \(\varepsilon>0\), we obtain from $r_{l}^{i}(\delta(\varepsilon))$ and $\mathcal{E}^i$ a perturbed realization-plan space, $\Lambda(\delta(\varepsilon))=\prod_{i\in N} \Lambda^i(\delta(\varepsilon))$, where
\begin{equation}\label{dflamda}\textstyle
	\Lambda^i(\delta(\varepsilon))=\left\{\gamma^i\in \Lambda^i\middle|\sum\limits_{\varpi^i\in E^i}\gamma^i(\varpi^i)\geq\sum\limits_{q=1}^{q^i(E^i)}r_q^i(\delta(\varepsilon))\mathrm{~for~all~}E^i\in \mathcal{E}^i\right\}.
\end{equation}
We illustrate $\Lambda^i(\delta(\varepsilon))$ with the game in Figure~\ref{fig:game1}. For ease of exposition, we focus on player~1 and use the representative subset \[\textstyle\tilde{\mathcal E}^1=\big\{
\{\varpi^1_{I^2_1}L_1\},\,\{\varpi^1_{I^1_1}C\},\,\{\varpi^1_{I^2_1}L_1,\varpi^1_{I^1_1}E\},\,\{\varpi^1_{I^2_1}L_1,\varpi^1_{I^3_1}L_2,\varpi^1_{I^1_1}E\}\big\}\subset \mathcal E^1.\] 
In contrast, $\{\varpi^1_{I^2_1}L_1,\varpi^1_{I^2_1}R_1\}\notin \mathcal E^1$ and $\{\varpi^1_{I^1_1}C,\varpi^1_{I^2_1}L_1\}\notin \mathcal E^1$, since they violate the two conditions in Definition~\ref{seqcons}, respectively. Given the perturbation terms  $r_q^1(\delta(\varepsilon))$, $q\in K^-_1$, the constraints in \eqref{dflamda} corresponding to $\tilde{\mathcal E}^1$ are
\[\begin{array}{ll}
	\gamma^1(\varpi^1_{I^2_1}L_1)\geq r_1^1(\delta(\varepsilon)),\!\!\!\!&\gamma^1(\varpi^1_{I^2_1}L_1)+\gamma^1(\varpi^1_{I^1_1}E)\geq
	\sum\limits_{q=1}^{2} r_q^1(\delta(\varepsilon)),\\
	\gamma^1(\varpi^1_{I^1_1}C)\geq
	\sum\limits_{q=1}^{3} r_q^1(\delta(\varepsilon)),\!\!\!\!&\gamma^1(\varpi^1_{I^2_1}L_1)+\gamma^1(\varpi^1_{I^3_1}L_2)+\gamma^1(\varpi^1_{I^1_1}E)\geq
	\sum\limits_{q=1}^{3} r_q^1(\delta(\varepsilon)).
\end{array}
\]
Appendix~\ref{compact} provides a polynomial-size extended formulation of $\Lambda^i(\delta(\varepsilon))$.

We then construct a perturbed game in sequence form, denoted by $\Gamma_s(\delta(\varepsilon))$, where player $i$'s best response to $\hat\gamma\in \Lambda(\delta(\varepsilon))$ is obtained by solving the linear optimization problem,
\begin{equation}
	\label{nfpr-opt-sfg1}
	\begin{array}{rl}
		\max\limits_{\gamma^i} &\sum\limits_{j\in M_i}\sum\limits_{a\in A(I^j_i)}\gamma^i(\varpi^i_{I^j_i}a)g^i(\varpi^i_{I^j_i}a,\hat\gamma^{-i})\\
		\text{s.t.}&\sum\limits_{a\in A(I^j_i)}\gamma^i(\varpi^i_{I^j_i}a)-\gamma^i(\varpi^i_{I^j_i})=0,\;j\in M_i,\\
		\quad&\sum\limits_{\varpi^i\in E^i}\gamma^i(\varpi^i)\geq\sum\limits_{q=1}^{q^i(E^i)}r_q^i(\delta(\varepsilon)),\;E^i\in \mathcal{E}^i.
	\end{array}
\end{equation}
We omit the payoff associated with the empty sequence in the objective function, as it does not depend on $\gamma^i$. We define $\gamma^*$ as a Nash equilibrium of $\Gamma_s(\delta(\varepsilon))$ precisely when $\gamma^{*i}$ individually solves Problem~(\ref{nfpr-opt-sfg1}) against $\gamma^{*-i}$ for every player $i\in N$. By applying the optimality conditions to Problem~(\ref{nfpr-opt-sfg1}) for all players and setting $\hat\gamma=\gamma$, we derive the polynomial equilibrium system,
\begingroup
\setlength{\jot}{-0.8pt}
\begin{equation}\label{nfpr-eqt-sfg1}\begin{array}{l}
		g^i(\varpi^i_{I^j_i}a,\gamma^{-i})\!+\!\sum\limits_{E^i\in \mathcal{E}^i(\varpi^i_{I^j_i}a)}\lambda^i(E^i)\!-\!\nu^i_{I^j_i}\!+\!\zeta^i_{I^j_i}(a)=0,i\in N,j\in M_i,a\in A(I^j_i),\\
		\sum\limits_{a\in A(I^j_i)}\gamma^i(\varpi^i_{I^j_i}a)-\gamma^i(\varpi^i_{I^j_i})=0,\;i\in N,j\in M_i,\\
		\Big(\sum\limits_{\varpi^i\in E^i}\gamma^i(\varpi^i)-\sum\limits_{q=1}^{q^i(E^i)}r_q^i(\delta(\varepsilon))\Big)\lambda^i(E^i)=0,\\
		\qquad\qquad\sum\limits_{\varpi^i\in E^i}\gamma^i(\varpi^i)\geq\sum\limits_{q=1}^{q^i(E^i)}r_q^i(\delta(\varepsilon)),\;\lambda^i(E^i)\geq0,\;i\in N,E^i\in \mathcal{E}^i,
	\end{array}
\end{equation}
\endgroup
where $\zeta^i_{I^j_i}(a)=\sum_{j_q\in M_i(\varpi^i_{I^j_i}a)}\nu^i_{I^{j_q}_i}$ and $\mathcal{E}^i(\varpi^i_{I^j_i}a)=\{E^i\in \mathcal{E}^i\mid\varpi^i_{I^j_i}a\in E^i\}$. Thus, $\gamma^*$ is a Nash equilibrium of $\Gamma_s(\delta(\varepsilon))$ if and only if $(\gamma^*,\lambda^*,\nu^*)$ satisfies System~\eqref{nfpr-eqt-sfg1} for some bounded pair $(\lambda^*,\nu^*)$.

\begin{theorem}\label{nfpr-the-sfc1} $\gamma^*\in\Lambda$ is a sequence-form proper equilibrium of $\Gamma$ if and only if there exists a positive sequence $\{\delta(\varepsilon^k)\}_{k=1}^{\infty}$ with $\lim_{k \to \infty} \varepsilon^k = 0$ and a sequence of realization plan profiles $\{\gamma(\varepsilon^k)\}_{k=1}^{\infty}$ such that, for each $k$, $\gamma(\varepsilon^k)$ is a Nash equilibrium of $\Gamma_s(\delta(\varepsilon^k))$ and $\lim_{k \to \infty} \gamma(\varepsilon^k) = \gamma^*$.
\end{theorem}
The proof of Theorem~\ref{nfpr-the-sfc1} is provided in Appendix~\ref{proofthm3}.

\section{Sequence-Form Differentiable Path-Following Methods for Computing Normal-Form Proper Equilibria}\label{nfpr-sec-prm4}
In this section, we use the preceding sequence-form characterization to develop two differentiable path-following methods for computing normal-form proper equilibria.
\subsection{A Logarithmic-Barrier Smooth Path}\label{nfpr-sct-log1}
Let $\gamma^0\in\Lambda_{++}$ be given and define $y^0=(y^{0i}(E^i):i\in N,E^i\in \mathcal{E}^i)$ by $y^{0i}(E^i)=\sum_{\varpi^i\in E^i}\gamma^{0i}(\varpi^i)$. Let $c(t) =\exp(1-1/t)$ for $t>0$ and $c(0)=0$. For $t\in [0, 1]$, let
$r_{l}^{i}(t)=(t/\omega_{0})^{\kappa_i(\varpi^i_\emptyset)-l},\,i\in N,\,l\in K^-_i$ with $\omega_0>1$. For $t\in(0,1]$, we construct a logarithmic-barrier artificial game $\Gamma_{s}^l(t)$, where each player $i\in N$'s best response to $\hat\gamma\in\Lambda$ is determined by solving the convex optimization problem,
\begingroup
\setlength{\jot}{-0.8pt}
\begin{equation}
	\label{nfpr-opt-log1}
	\begin{array}{rl}
		\max\limits_{\gamma^i} & (1-c(t))\sum\limits_{j\in M_i}\sum\limits_{a\in A(I^j_i)}\gamma^i(\varpi^i_{I^j_i}a)g^i(\varpi^i_{I^j_i}a,\hat\gamma^{-i})\\
		&+c(t)\sum\limits_{E^i\in \mathcal{E}^i}\Big(y^{0i}(E^i)\ln y^i(E^i)-\sum\limits_{\varpi^i\in E^i}\gamma^i(\varpi^i)\Big)\\
		\text{s.t.}& \sum\limits_{\varpi^i\in E^i}\gamma^i(\varpi^i)-\sum\limits_{q=1}^{q^i(E^i)}r_q^i(t)(1-t)-y^i(E^i)=0,\; E^i\in \mathcal{E}^i,\\
		& \sum\limits_{a\in A(I^j_i)}\gamma^i(\varpi^i_{I^j_i}a)-\gamma^i(\varpi^i_{I^j_i})=0,\;j\in M_i.
	\end{array}
\end{equation}
\endgroup
Applying the optimality conditions to Problem~(\ref{nfpr-opt-log1}) and imposing the fixed-point condition $\hat{\gamma}=\gamma$, we obtain the polynomial equilibrium system associated with $\Gamma_{s}^l(t)$,
\begin{equation}\label{nfpr-eqt-log1}\begin{array}{l}
		(1-c(t))g^i(\varpi^i_{I^j_i}a,\gamma^{-i})+\sum\limits_{E^i\in \mathcal{E}^i(\varpi^i_{I^j_i}a)}\lambda^i(E^i)-c(t)|\mathcal{E}^i(\varpi^i_{I^j_i}a)|\\
		\hspace{4.5cm}-\nu^i_{I^j_i}+\zeta^i_{I^j_i}(a)=0,\;i\in N,j\in M_i,a\in A(I^j_i),\\
		\sum\limits_{\varpi^i\in E^i}\gamma^i(\varpi^i)-\sum\limits_{q=1}^{q^i(E^i)}r_q^i(t)(1-t)-y^i(E^i)=0,\; i\in N, E^i\in \mathcal{E}^i,\\
		\sum\limits_{a\in A(I^j_i)}\gamma^i(\varpi^i_{I^j_i}a)-\gamma^i(\varpi^i_{I^j_i})=0,\;i\in N,j\in M_i,\\
		y^i(E^i)\lambda^i(E^i)=c(t)y^{0i}(E^i),\; y^i(E^i)>0,\; i\in N,E^i\in \mathcal{E}^i,
	\end{array}
\end{equation}
where $\zeta^i_{I^j_i}(a)=\sum_{j_q\in M_i(\varpi^i_{I^j_i}a)}\nu^i_{I^{j_q}_i}$. It can be observed that $\gamma^*$ solves Problem~(\ref{nfpr-opt-log1}) against itself if and only if there exists a unique $(y^*,\lambda^*, \nu^*)$ such that, along with $\gamma^*$, they fulfill System~(\ref{nfpr-eqt-log1}).

We next demonstrate the existence of a smooth path along which the points satisfy System~(\ref{nfpr-eqt-log1}). Moreover, this path admits a unique starting point and converges to a sequence-form proper equilibrium. A rigorous proof follows the approach in \cite{houSequenceFormCharacterizationDifferentiable2025}.
\begin{itemize}
	\item \textbf{Uniqueness of the Starting Point.} At $t=1$, System~(\ref{nfpr-eqt-log1}) has a unique solution, denoted by $(\gamma^*(1),y^*(1),\lambda^*(1),\nu^*(1))$, whose components satisfy $\gamma^{*i}(1;\varpi^i_{I^j_i}a)=\gamma^{0i}(\varpi^i_{I^j_i}a)$, $y^{*i}(E^i)=y^{0i}(E^i)$, $\lambda^{*i}(1;\varpi^{i}_{I^j_i}a)=1$, and $\nu^{*i}_{I^j_i}(1)=0$ for $i\in N,j\in M_i,a\in A(I^j_i)$.
	
	\item \textbf{Convergence Analysis}. Considering $\Gamma_{s}(\delta(\varepsilon))$ with $r_{l}^{i}(\delta(\varepsilon))=(1-t)r_{l}^{i}(t)$ for $i \in N$ and $l \in K_i$. The scaling term $(1-t)$ ensures nonemptiness of the feasible strategy space when $t$ is not sufficiently small. This specification leads to System~\eqref{nfpr-eqt-sfg1}. By applying Theorem 2.2 in \cite{luoExtensionHoffmansError1994}, we obtain that every limit point of $\{\gamma^*(t_k)\}_{k=1}^{\infty}$ yields a sequence-form proper equilibrium, where each $\gamma^*(t_k)$ is a Nash equilibrium of $\Gamma_s^l(t_k)$ with $\lim_{k\to \infty}t_k=0$.
	
	\item \textbf{Feasible Connected Component}. Let $\mathscr{C}_L$ denote the closure of the solution set of System~(\ref{nfpr-eqt-log1}) for $0<t\leq 1$, whose boundedness can be proved. By exploiting Browder's (\citeyear{Browdercontinuityfixedpoints1960}) fixed-point theorem, we can prove that there exists a connected component in $\mathscr{C}_L$ intersecting both the $t=1$ and $t=0$ levels, thereby yielding a continuous solution path extends from $t=1$ to $t=0$, without premature termination or divergence.
	
	\item \textbf{Existence of Smooth Path}. Let ${\bm\alpha}=(\alpha(\varpi^i_{I^j_i}a):i\in N, j\in M_i, a\in A(I^j_i))^\top\in\mathbb{R}^{n_0}$ with sufficiently small $\|\bm\alpha\|$. Subtracting $c(t)(1-t)\bm\alpha$ from the first group of equations in \eqref{nfpr-eqt-log1}, we obtain an approximately equivalent system with equalities written as $p(\gamma,\lambda,\nu,t;\bm\alpha)=0$ that preserves the limiting equilibrium conclusions. It can be shown that the Jacobian matrices of $p(\gamma,\lambda,\nu,t;\bm\alpha)$ and $p_{\bm\alpha}(\gamma,\lambda,\nu,1;\bm\alpha)$ are of full row rank. Consequently, the transversality theorem~\citep{EavesGeneralequilibriummodels1999} and the implicit function theorem establish that, for almost any ${\bm\alpha}\in\mathbb{R}^{n_0}$ with sufficiently small $\|{\bm\alpha}\|$, the continuous solution path previously identified is in fact smooth. The preceding analysis further shows that this smooth path is unique and connects $(\gamma^*(1),\lambda^*(1),\nu^*(1),1)$ at $t=1$ to a sequence-form proper equilibrium as $t\to0$.
\end{itemize}

To eliminate the inequality constraints in (\ref{nfpr-eqt-log1}), we employ a variable substitution as outlined in~\cite{Caodifferentiablepathfollowingmethod2022}. Let $\tau_0>0$ and $\kappa_0>1$, and introduce the functions $\psi_1(w,r;\tau_0,\kappa_0)=\left((w+\sqrt{w^2+4\tau_0r})/2\right)^{\kappa_0}$ and $\psi_2(w,r;\tau_0,\kappa_0)=\left((-w+\sqrt{w^2+4\tau_0r})/2\right)^{\kappa_0}$, satisfying $\psi_1(w,r;\tau_0,\kappa_0)\psi_2(w,r;\tau_0,\kappa_0)=(\tau_0r)^{\kappa_0}$. The condition $\kappa_0>1$ ensures continuous differentiability of both functions over $\mathbb{R}\times(0,\infty)$. For $x=(x^i(E^i):i\in N,E^i\in \mathcal{E}^i)\in \mathbb{R}^{e_0}$, define the mappings $y(x,t)=(y^i(x,t;E^i):i\in N,E^i\in \mathcal{E}^i)$ and $\lambda(x,t)=(\lambda^i(x,t;E^i):i\in N,E^i\in \mathcal{E}^i)$ by $y^i(x,t;E^i)=\psi_1(x^i(E^i),c(t)^{1/\kappa_0}; y^{0i}(E^i)^{1/\kappa_0}, \kappa_0)$ and $\lambda^i(x,t;E^i)=\psi_2(x^i(E^i),c(t)^{1/\kappa_0};y^{0i}(E^i)^{1/\kappa_0},\kappa_0),\;i\in N,E^i\in \mathcal{E}^i$,
so that $y^i(x,t;E^i)\lambda^i(x,t;E^i)=c(t)y^{0i}(E^i)$ holds for $i\in N,E^i\in \mathcal{E}^i$. By substituting $y^i(E^i)$ and $\lambda^i(E^i)$ with $y^i(x,t;E^i)$ and $ \lambda^i(x,t;E^i)$ in System~(\ref{nfpr-eqt-log1}), we obtain an equivalent formulation free of inequality constraints,
\begin{equation}\label{nfpr-eqt-log4}\begin{array}{l}
		(1-c(t))g^i(\varpi^i_{I^j_i}a,\gamma^{-i})+\sum\limits_{E^i\in \mathcal{E}^i(\varpi^i_{I^j_i}a)}\lambda^i(x,t;E^i)-c(t)|\mathcal{E}^i(\varpi^i_{I^j_i}a)|\\
		\hspace{0.5cm}-\nu^i_{I^j_i}+\zeta^i_{I^j_i}(a)-c(t)(1-t)\alpha(\varpi^i_{I^j_i}a)=0,\; i\in N,j\in M_i,a\in A(I^j_i),\\
		\sum\limits_{\varpi^i\in E^i}\gamma^i(\varpi^i)-\sum\limits_{q=1}^{q^i(E^i)}r_q^i(t)(1-t)-y^i(x,t;E^i)=0,\; i\in N,E^i\in \mathcal{E}^i,\\
		\sum\limits_{a\in A(I^j_i)}\gamma^i(\varpi^i_{I^j_i}a)-\gamma^i(\varpi^i_{I^j_i})=0,\; i\in N,j\in M_i.
	\end{array}
\end{equation}

At $t=1$, System~(\ref{nfpr-eqt-log4}) has a unique solution $(\gamma^*(1),x^*(1),\nu^*(1),1)$ with $\gamma^{*i}(1;\varpi^i_{I^j_i}a)=\gamma^{0i}(\varpi^i_{I^j_i}a)$, $x^{*i}(1;E^i)=y^{0i}(E^i)^{1/\kappa_0}-1$, and $\nu^{*i}_{I^j_i}(1)=0$ for $i\in N,j\in M_i,a\in A(I^j_i),E^i\in \mathcal{E}^i$. Define $\widetilde{\mathscr{P}}_L=\{(\gamma,x,\nu,t)\mid (\gamma,x,\nu,t)$ satisfies System~(\ref{nfpr-eqt-log4}) with $0<t\leq 1\}$ and let $\mathscr{P}_L$ be its closure. The preceding discussion reveals that $\mathscr{P}_{L}$ contains a smooth path that originates from the point $(\gamma^*(1), x^*(1),\nu^*(1),1)$ and converges to a sequence-form proper equilibrium as the parameter $t$ approaches $0$. Then, Normal-form proper equilibria can be obtained either by applying the strategy correspondence $T$ to limit points of the realization-plan path, or as limit points of the mixed-strategypath induced by any continuous strategy transformation consistent with $T$, such as $\sigma(\cdot)$.

\subsection{An Entropy-Barrier Smooth Path}
This subsection introduces an alternative computational method by constructing an entropy-barrier artificial game $\Gamma_s^e(t)$. In this game, each player $i$ determines an optimal response to a given strategy profile $\hat\gamma\in\Lambda$ by solving the convex optimization problem,
\begin{equation}
	\label{nfpr-opt-ent1}
	\begin{array}{rl}
		\max\limits_{\gamma^i} & (1-c(t))\sum\limits_{j\in M_i}\sum\limits_{a\in A(I^j_i)}\gamma^i(\varpi^i_{I^j_i}a)g^i(\varpi^i_{I^j_i}a,\hat\gamma^{-i})\\
		&-c(t)\sum\limits_{E^i\in \mathcal{E}^i}y^i(E^i)\big(\ln y^i(E^i)-\ln y^{0i}(E^i)-1\big)\\
		\text{s.t.} & \sum\limits_{\varpi^i\in E^i}\gamma^i(\varpi^i)-\sum\limits_{q=1}^{q^i(E^i)}r_q^i(t)(1-t)-y^i(E^i)=0,\; E^i\in \mathcal{E}^i,\\
		& \sum\limits_{a\in A(I^j_i)}\gamma^i(\varpi^i_{I^j_i}a)-\gamma^i(\varpi^i_{I^j_i})=0,\;j\in M_i.
	\end{array}
\end{equation}

Through the application of the optimality conditions to Problem~(\ref{nfpr-opt-ent1}), and under the setting $\hat{\gamma} = \gamma$, the equilibrium system associated with $\Gamma_{s}^e(t)$ is derived as
\begin{equation}\label{nfpr-eqt-ent1}\begin{array}{l}
		(1-c(t))g^i(\varpi^i_{I^j_i}a,\gamma^{-i})-c(t)\sum\limits_{E^i\in \mathcal{E}^i(\varpi^i_{I^j_i}a)}\big(\ln y^i(E^i)-\ln y^{0i}(E^i)\big)\\
		\hspace{3.7cm}-\nu^i_{I^j_i}+\zeta^i_{I^j_i}(a)=0,\;i\in N,j\in M_i,a\in A(I^j_i),\\
		\sum\limits_{\varpi^i\in E^i}\gamma^i(\varpi^i)-\sum\limits_{q=1}^{q^i(E^i)}r_q^i(t)(1-t)-y^i(E^i)=0,\; i\in N, E^i\in \mathcal{E}^i,\\
		\sum\limits_{a\in A(I^j_i)}\gamma^i(\varpi^i_{I^j_i}a)-\gamma^i(\varpi^i_{I^j_i})=0,\;i\in N,j\in M_i,\\
	\end{array}
\end{equation}
where $\zeta^i_{I^j_i}(a)=\sum_{j_q\in M_i(\varpi^i_{I^j_i}a)}\nu^i_{I^{j_q}_i}$. The logarithmic terms in System~(\ref{nfpr-eqt-ent1}) may cause numerical difficulties, since their derivatives grow sharply as the arguments approach zero, potentially leading to computational instability. To circumvent this issue, we introduce the following transformation. Let $\rho_0=\max_{i\in N,E^i\in \mathcal{E}^i}|E^i|$ and $b=(b^i(E^i):i\in N,E^i\in \mathcal{E}^i)\in \mathbb{R}^{e_0}$. For $i\in N,E^i\in \mathcal{E}^i$, we define $y^i(b;E^i) = \rho_0\phi(b^i(E^i))$, where the function $\phi:\mathbb{R}\to\mathbb{R}$ is given by
\[\textstyle
\phi(v) =
\begin{cases}
	e^{1 - \frac{1}{v}}, & \text{if } v > 0, \\
	0, & \text{if } v \leq 0,
\end{cases}
\text{ with derivative }
\frac{d}{dv} \phi(v) =
\begin{cases}
	\frac{e^{1 - \frac{1}{v}}}{v^2}, & \text{if } v > 0, \\
	0, & \text{if } v \leq 0.
\end{cases}
\]
Hence, $y^i(b;E^i)$ is smooth for $b^i(E^i)>0$ and vanishes otherwise. Replacing $y^i(E^i)$ by $y^i(b;E^i)$ and setting $\lambda^i(E^i)=c(t)/b^i(E^i)$ for $i\in N,E^i\in \mathcal{E}^i$ yield an equivalent reformulation of System~(\ref{nfpr-eqt-ent1}). Further incorporating the adjustment term $-c(t)(1-t){\bm\alpha}$, we obtain
\begin{equation}\label{nfpr-eqt-ent3}\begin{array}{ll}
		(1-c(t))g^i(\varpi^i_{I^j_i}a,\gamma^{-i})+\sum\limits_{E^i\in \mathcal{E}^i(\varpi^i_{I^j_i}a)}\big(\lambda^i(E^i)+c(t)\ln y^{0i}(E^i)\big)-\nu^i_{I^j_i}+\zeta^i_{I^j_i}(a)\\
		\!-\! c(t)(\ln\rho_0+1)|\mathcal{E}^i(\varpi^i_{I^j_i}a)|\!-\! c(t)(1-t)\alpha(\varpi^i_{I^j_i}a)\!=\!0,i\in N,j\in M_i,a\in A(I^j_i),\\
		\sum\limits_{\varpi^i\in E^i}\gamma^i(\varpi^i)-\delta_0\sum\limits_{q=1}^{q^i(E^i)}r_q^i(t)(1-t)-y^i(b;E^i)=0,\; i\in N, E^i\in \mathcal{E}^i,\\
		\sum\limits_{a\in A(I^j_i)}\gamma^i(\varpi^i_{I^j_i}a)-\gamma^i(\varpi^i_{I^j_i})=0,\;i\in N,j\in M_i,\\
		\lambda^i(E^i)b^i(E^i)=c(t),\;\lambda^i(E^i)>0,\;i\in N,E^i\in \mathcal{E}^i.\\
	\end{array}
\end{equation}
We perform a variable substitution to eliminate the inequality constraints. For $x=(x^i(E^i):i\in N,E^i\in\mathcal{E}^i)\in \mathbb{R}^{e_0}$, define $b(x,t)=(b^i(x,t;E^i):i\in N,E^i\in\mathcal{E}^i)$ and $\lambda(x,t)=(\lambda^i(x,t;E^i):i\in N,E^i\in\mathcal{E}^i)$,  where $b^i(x,t;E^i)=\psi_1(x^i(E^i),c(t)^{1/\kappa_0}; 1, \kappa_0)$, $\lambda^i(x,t;E^i)=\psi_2(x^i(E^i),c(t)^{1/\kappa_0};1,\kappa_0),\;i\in N,E^i\in \mathcal{E}^i$.
It is immediate that $b^i(x,t;E^i)\lambda^i(x,t;E^i)=c(t)$ for $i\in N,E^i\in \mathcal{E}^i$. Substituting $b^i(E^i)$ and $\lambda^i(E^i)$ in System~(\ref{nfpr-eqt-ent3}) with $b^i(x,t;E^i)$ and $\lambda^i(x,t;E^i)$ for $i\in N,E^i\in \mathcal{E}^i$, we obtain an equivalent formulation,
\begin{equation}\label{nfpr-eqt-ent4}\begin{array}{l}
		(1-c(t))g^i(\varpi^i_{I^j_i}a,\gamma^{-i})\!+\!\sum\limits_{E^i\in \mathcal{E}^i(\varpi^i_{I^j_i}a)}\big(\lambda^i(x,t;E^i)\!+\!c(t)\ln y^{0i}(E^i)\big)\!-\!\nu^i_{I^j_i}\!+\!\zeta^i_{I^j_i}(a)\\
		\!-\! c(t)(\ln\rho_0\!+\! 1)|\mathcal{E}^i(\varpi^i_{I^j_i}a)|\!-\! c(t)(1\!-\! t)\alpha(\varpi^i_{I^j_i}a)=0,\;i\in N,j\in M_i,a\in A(I^j_i),\\
		\sum\limits_{\varpi^i\in E^i}\gamma^i(\varpi^i)-\delta_0\sum\limits_{q=1}^{q^i(E^i)}r_q^i(t)(1-t)-y^i(b(x,t);E^i)=0,\; i\in N, E^i\in \mathcal{E}^i,\\
		\sum\limits_{a\in A(I^j_i)}\gamma^i(\varpi^i_{I^j_i}a)-\gamma^i(\varpi^i_{I^j_i})=0,\;i\in N,j\in M_i.\\
	\end{array}
\end{equation}

When $t=1$, System~(\ref{nfpr-eqt-ent4}) admits a unique solution \( (\gamma^*(1),x^*(1),\nu^*(1)) \), where $\gamma^{*i}(1;\varpi^i_{I^j_i}a)=\gamma^{0i}(\varpi^i_{I^j_i}a)$, $x^{*i}(1;E^i)=(\ln\rho_0+1-\ln y^{0i}(E^i))^{-1/\kappa_0}-(\ln\rho_0+1-\ln y^{0i}(E^i))^{1/\kappa_0}$, and $\nu^{*i}_{I^j_i}(1)=0$ for $i\in N,j\in M_i,a\in A(I^j_i),E^i\in \mathcal{E}^i$. Following the approach of the previous subsection, we can show that System~(\ref{nfpr-eqt-ent4}) specifies a smooth path originating from $(\gamma^*(1),x^*(1),\nu^*(1))$ at $t=1$ and converging to a sequence-form proper equilibrium as $t\to0$. By further exploiting the correspondence $T$, one can derive one or more normal-form proper equilibria.

\section{Numerical Performance}\label{nfpr-sec-prm5}

In this section, we present a set of numerical experiments aimed at evaluating the effectiveness and efficiency of the proposed methods. For the numerical implementation, we employed a predictor--corrector method in \cite{AllgowerNumericalContinuationMethods1990} and \cite{EavesGeneralequilibriummodels1999} to trace the smooth paths defined by Systems (\ref{nfpr-eqt-log4}) and (\ref{nfpr-eqt-ent4}), referred to as LGPR and ETPR, respectively. At each iteration, the predictor was used to generate an approximation to the next point on the path, and the corrector was then applied to refine this approximation. In our experiments, both the required accuracy of the starting point for each corrector step and the predictor step size were set to $0.2t^{0.3}$. A run was declared successful when $t<10^{-4}$; otherwise, it was terminated if the number of iterations or the computation time exceeded the prescribed limits. All experiments were conducted on a Windows Server 2016 Standard machine equipped with two Intel(R) Xeon(R) CPU E5-2650 v4 @ 2.20GHz processors and 128 GB of RAM.
\begin{example}\label{nfpr-exm-num1} This example examines the convergence behavior of the proposed methods in two extensive-form games. We apply the strategy transformation $\sigma(\cdot)$ to the computed realization-plan paths and obtain the corresponding mixed-strategy paths, which provide a direct visualization in the reduced normal form. For the game in Figure~\ref{fig:game1}, the four panels in Figure~\ref{fig:test5} compare the mixed-strategy paths generated by different methods under different starting points. These paths show that different initializations may lead the proposed methods to distinct normal-form proper equilibria. We further consider the three-player game in Figure~\ref{fig:game2} and the four-player game in Figure~\ref{fig:game3}. The representative mixed-strategy paths shown in Figures~\ref{fig:game2-ms-paths} and~\ref{fig:game3-ms-paths} illustrate convergence in the multi-player setting. For the three-player game, the paths converge to the normal-form proper equilibrium
	$\bm{\sigma}^1=(0,\tfrac{24}{49},\tfrac{25}{49})^\top$,
	$\bm{\sigma}^2=(\tfrac{3}{8},\tfrac{5}{8})^\top$, and
	$\bm{\sigma}^3=(\tfrac{1}{4},\tfrac{3}{4})^\top$.
	For the four-player game, they converge to the normal-form proper equilibrium
	$\bm{\sigma}^1=(0,1,0)^\top$, $\bm{\sigma}^2=(1,0)^\top$, $\bm{\sigma}^3=(1,0)^\top$, and $\bm{\sigma}^4=(0,0,1,0)^\top$.
	
	\begin{figure}[htbp]
		\centering
		\begin{minipage}[b]{0.46\textwidth}
			\vspace{0pt}
			\centering
			\includegraphics[width=0.95\textwidth]{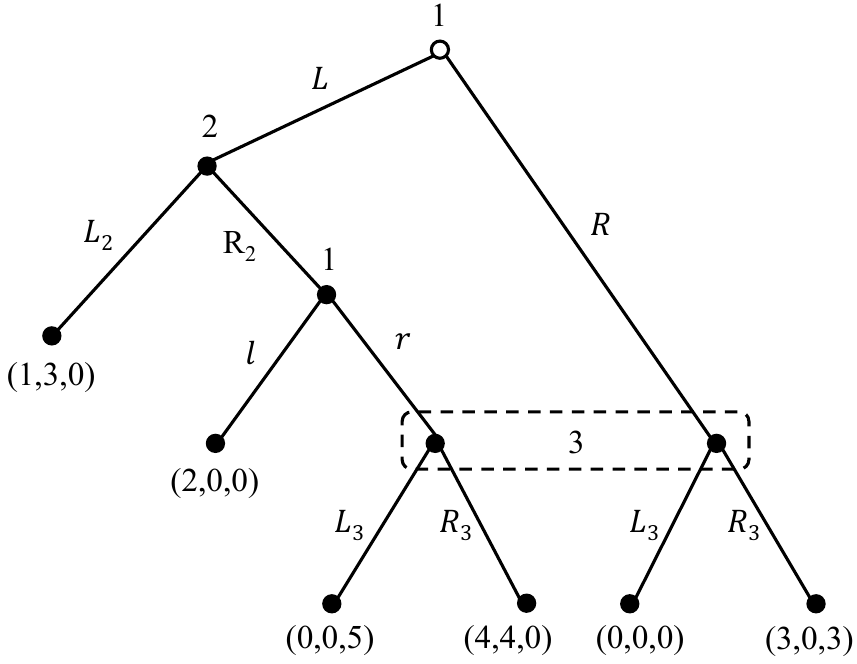}
			\captionof{figure}{A Three-Player Game from~\cite{SeltenReexaminationperfectnessconcept1975}}
			\label{fig:game2}
		\end{minipage}%
		\begin{minipage}[b]{0.54\textwidth}
			\vspace{0pt}
			\centering
			\includegraphics[width=0.9\textwidth]{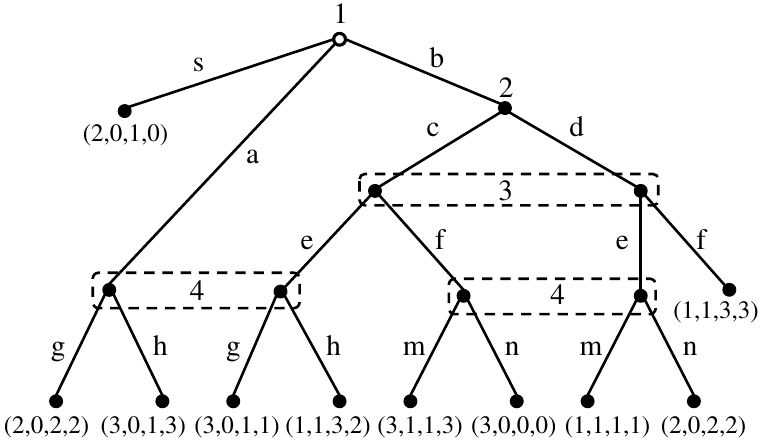}
			\captionof{figure}{A Four-Player Game from~\cite{bonannoAGMconsistencyPerfectBayesian2013}}
			\label{fig:game3}
		\end{minipage}
		
	\end{figure}
	
\end{example}

\begin{figure}[htbp]
	\centering
	\begin{tabular}{cc}
		\subcaptionbox{LGPR\label{NFPR-fig-pth1}}
		{\includegraphics[width=0.42\textwidth]{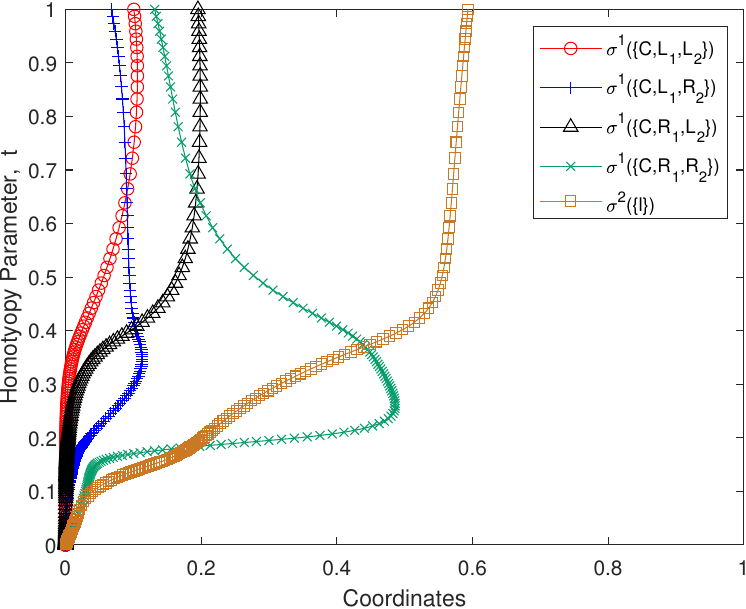}}
		\hspace{0.7cm}
		&
		\subcaptionbox{ETPR\label{NFPR-fig-pth2}}
		{\includegraphics[width=0.42\textwidth]{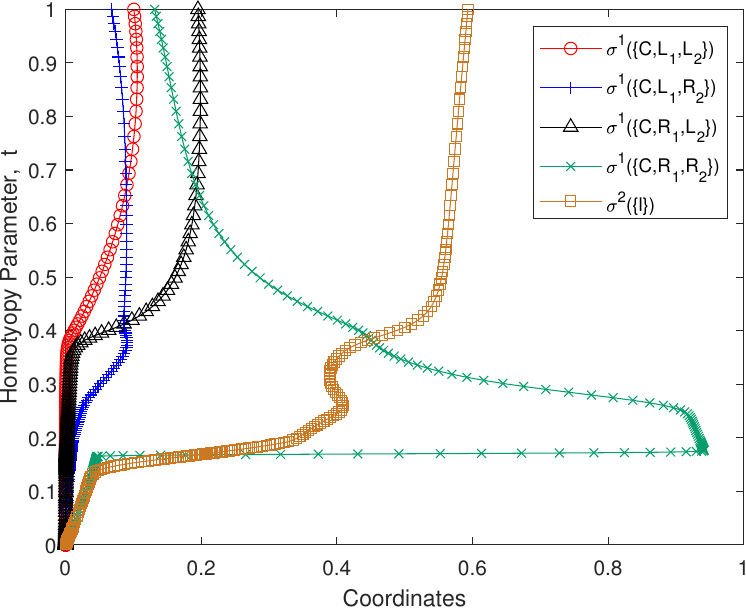}}
		\\[1ex]
		
		\subcaptionbox{LGPR\label{NFPR-fig-pth3}}
		{\includegraphics[width=0.42\textwidth]{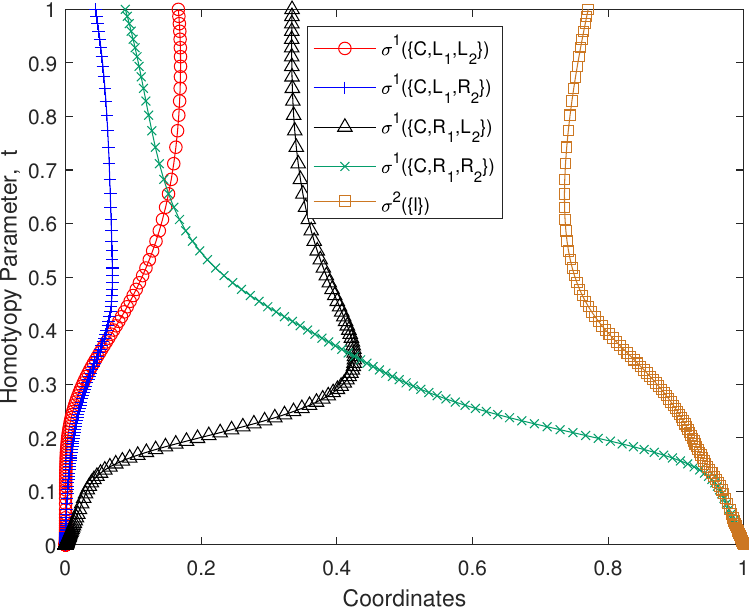}}
		\hspace{0.7cm}
		&
		\subcaptionbox{ETPR\label{NFPR-fig-pth4}}
		{\includegraphics[width=0.42\textwidth]{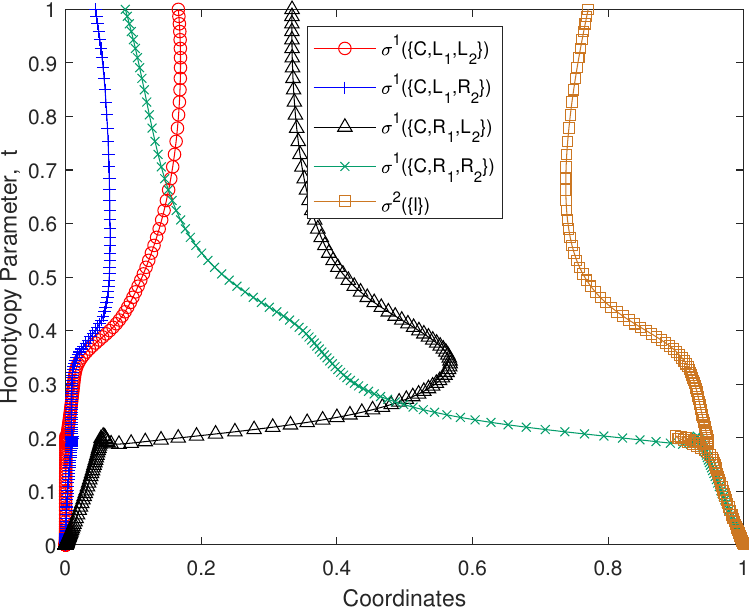}}
	\end{tabular}
	
	\caption{Mixed-Strategy Paths from Distinct Starting Points for the Game in Figure~\ref{fig:game1}}
	\label{fig:test5}
\end{figure}

\begin{figure}[htbp]
	\centering
	
	\subcaptionbox{LGPR\label{fig:game2-lgpr}}
	{\includegraphics[width=0.42\textwidth]{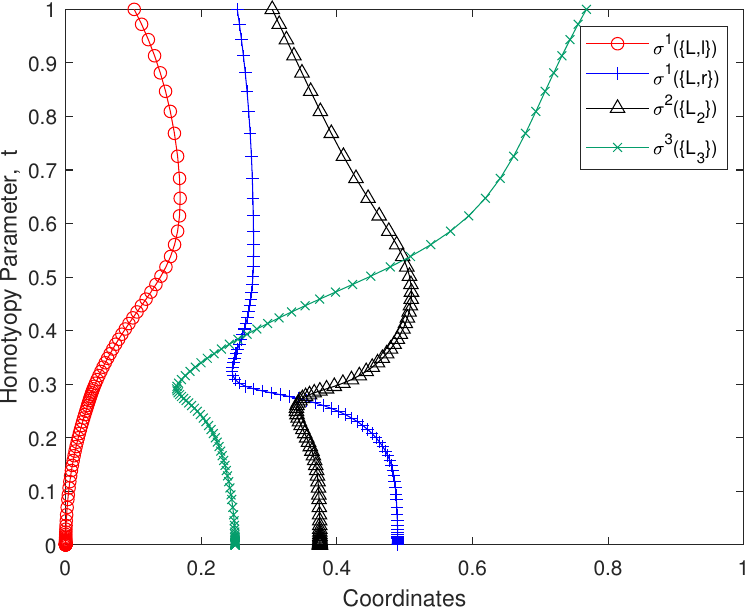}}
	\hspace{0.7cm}
	\subcaptionbox{ETPR\label{fig:game2-etpr}}
	{\includegraphics[width=0.42\textwidth]{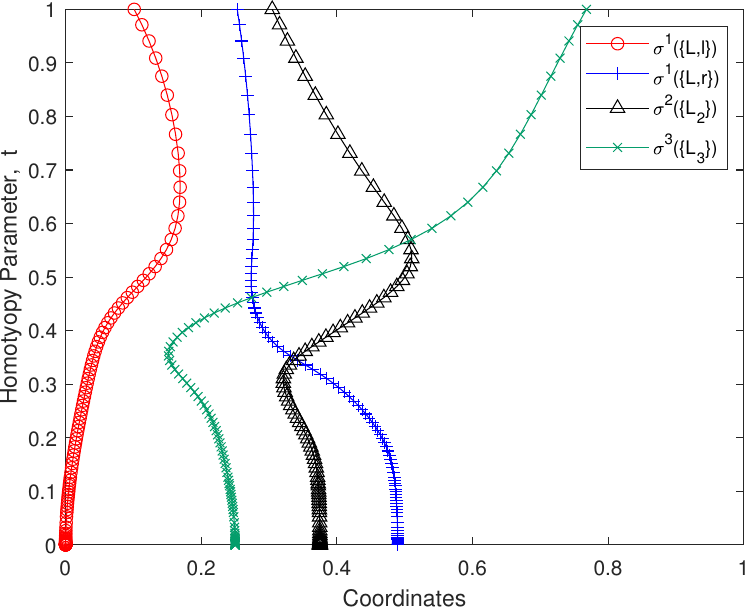}}
	
	\caption{Mixed-Strategy Paths for the Game in Figure~\ref{fig:game2}.}
	\label{fig:game2-ms-paths}
\end{figure}
\begin{figure}[htbp]
	\centering
	
	\subcaptionbox{LGPR\label{fig:game3-lgpr}}
	{\includegraphics[width=0.42\textwidth]{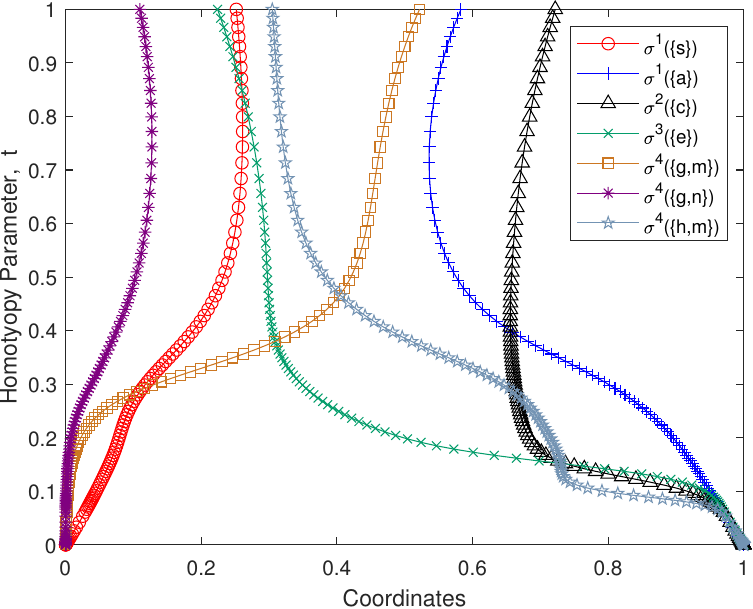}}
	\hspace{0.7cm}
	\subcaptionbox{ETPR\label{fig:game3-etpr}}
	{\includegraphics[width=0.42\textwidth]{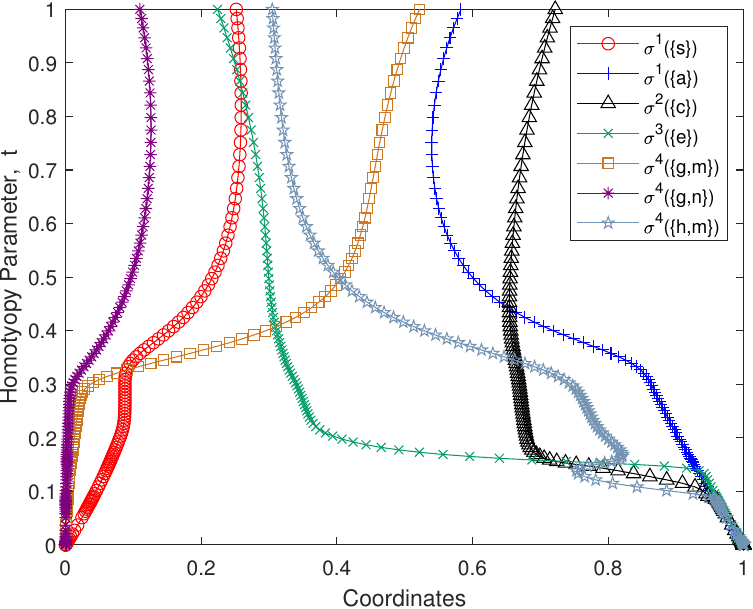}}
	
	\caption{Mixed-Strategy Paths for the Game in Figure~\ref{fig:game3}.}
	\label{fig:game3-ms-paths}
\end{figure}
\vspace{-24pt}
\begin{example}\label{nfpr-exm-num3} To evaluate the computational performance of our methods, we consider two structurally distinct types of random extensive-form games from~\cite{houSequenceFormCharacterizationDifferentiable2025}, depicted in Figures~\ref{nfpr-fig-gam5}--\ref{nfpr-fig-gam6}. In both types, players move cyclically, and terminal payoffs are drawn independently from the uniform distribution on $[-10,10]$. Each game is parameterized by the number of players ($n$), the maximum historical depth ($\mathcal{L}$), and the number of available actions ($\mathcal{A}$) at each information set. For each type and parameter configuration $(n,\mathcal{L}, \mathcal{A})$, $20$ random game  instances with distinct payoffs were generated and subsequently solved by the proposed methods. For each instance, both methods were initialized from the same randomly generated starting point, and the predictor-corrector parameters were kept unchanged across all experiments.
	\begin{figure}[htbp]
		\centering
		\begin{minipage}[b]{0.5\textwidth}
			\vspace{0pt}
			\centering
			\includegraphics[width=0.9\textwidth]{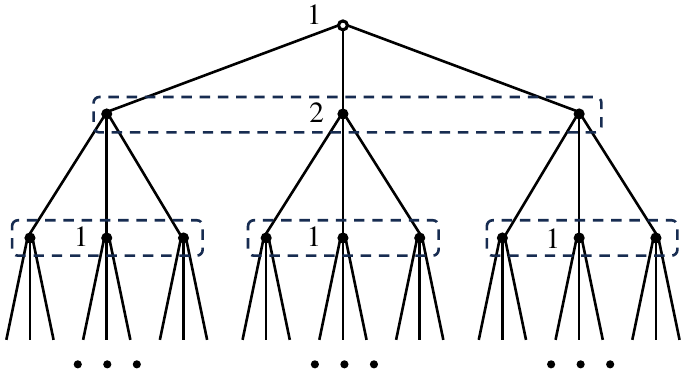}
			\captionof{figure}{Type-1 Game}
			\label{nfpr-fig-gam5}
		\end{minipage}%
		\begin{minipage}[b]{0.48\textwidth}
			\vspace{0pt}
			\centering
			\includegraphics[width=0.8\textwidth]{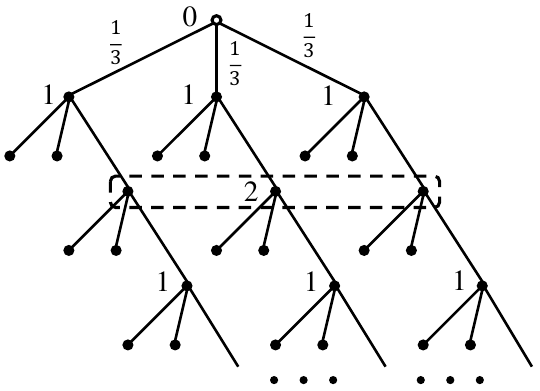}
			\captionof{figure}{Type-2 Game}
			\label{nfpr-fig-gam6}
		\end{minipage}
		
	\end{figure}
	
	Direct normal-form path-following methods are not included in the comparison, since the normal-form representation grows exponentially and quickly becomes computationally impractical, rendering such methods an uninformative benchmark for the sequence-form methods developed here. For example, in Type 1 with parameter setting $(3,3,3)$, player~3 already has $27$ pure strategies. Overall, the numerical results in Tables~\ref{t2}-\ref{t3} demonstrate that both LGPR and ETPR exhibit good performance across a range of game sizes, with no failures observed in most small and medium instances. LGPR generally achieves fewer iteration numbers and less computational time, indicating better numerical efficiency. By contrast, for larger instances such as $(3,7,2)$ and $(3,4,3)$ in Table~\ref{t3}, ETPR more often fails to reach $t<10^{-4}$ within the prescribed iteration or time limits. This suggests that LGPR is more reliable under a fixed computational budget, while ETPR remains of practical interest as an alternative method due to its different convergence path.
\end{example}
\begin{table}[htbp]
	\centering
	\caption{Numerical Comparisons for the Game in Figure~\ref{nfpr-fig-gam5}}
	\label{t2}
	\fontsize{10pt}{11pt}\selectfont
	\setlength{\tabcolsep}{3pt}
	\begin{tabular*}{\textwidth}{
			@{\extracolsep{\fill}}
			lllllllll
		}
		\toprule
		& & & \multicolumn{2}{l}{Iteration Numbers}
		& \multicolumn{2}{l}{Computational Time}
		& \multicolumn{2}{l}{Failure Rates} \\
		\cmidrule(r){4-5}\cmidrule(r){6-7}\cmidrule(r){8-9}
		
		$(n,\mathcal{L},\mathcal{A})$
		& Dim & & LGPR & ETPR & LGPR & ETPR & LGPR & ETPR \\
		\midrule
		
		$(2,3,2)$ & $21$ & max & 139 & 197 & 1.0 & 1.8 & 0\% & 0\% \\
		& & min & 85 & 120 & 0.6 & 1.0 & & \\
		& & \textbf{med} & \textbf{112.5} & \textbf{155.0}
		& \textbf{0.8} & \textbf{1.3} & & \\
		
		$(2,4,2)$ & $121$ & max & 333 & 549 & 20.0 & 39.4 & 0\% & 0\% \\
		& & min & 131 & 177 & 7.7 & 12.4 & & \\
		& & \textbf{med} & \textbf{195.5} & \textbf{286.5}
		& \textbf{11.4} & \textbf{18.7} & & \\
		
		$(3,4,2)$ & $117$ & max & 512 & 544 & 31.6 & 48.4 & 0\% & 0\% \\
		& & min & 179 & 282 & 13.2 & 23.0 & & \\
		& & \textbf{med} & \textbf{199.0} & \textbf{329.0}
		& \textbf{15.4} & \textbf{26.3} & & \\
		
		$(3,3,3)$ & $360$ & max & 628 & 909 & 253.3 & 635.6 & 0\% & 0\% \\
		& & min & 240 & 333 & 113.0 & 235.2 & & \\
		& & \textbf{med} & \textbf{328.0} & \textbf{481.5}
		& \textbf{147.5} & \textbf{327.7} & & \\
		
		$(2,3,3)$ & $522$ & max & 1306 & -- & 3487.6 & -- & 0\% & 5\% \\
		& & min & 424 & 299 & 1256.9 & 1836.7 & & \\
		& & \textbf{med} & \textbf{546.5} & \textbf{421.0}
		& \textbf{1560.9} & \textbf{2516.1} & & \\
		\bottomrule
	\end{tabular*}
\end{table}

\section{Conclusion}\label{nfpr-sec-prm6}

This paper has developed a unified framework in sequence form for characterizing and computing normal-form proper equilibria in finite extensive-form games with perfect recall. We have introduced sequence-form proper equilibrium by defining payoff comparisons over sequences, and have established its equivalence to normal-form proper equilibrium. We have further presented a class of perturbed sequence-form games by imposing realization probability lower bounds on restricted sequence sets, with the bounds specified by the perturbation hierarchy, thereby yielding an equivalent computational formulation. Leveraging this formulation, we have proposed two differentiable path-following methods, based respectively on the logarithmic-barrier game and the entropy-barrier game, that generate smooth equilibrium paths converging to normal-form proper equilibria. Numerical experiments have confirmed the effectiveness and computational efficiency of the proposed methods. Future research may extend this framework investigate path-following methods with selection properties for computing more refined normal-form proper equilibria.
\begin{table}[htbp]
	\centering
	\caption{Numerical Comparisons for the Game in Figure~\ref{nfpr-fig-gam6}}
	\label{t3}
	\fontsize{10pt}{11pt}\selectfont
	\setlength{\tabcolsep}{3pt}
	\begin{tabular*}{\textwidth}{
			@{\extracolsep{\fill}}
			lllllllll
		}
		\toprule
		& & &
		\multicolumn{2}{l}{Iteration Numbers}
		& \multicolumn{2}{l}{Computational Time}
		& \multicolumn{2}{l}{Failure Rates} \\
		\cmidrule(r){4-5}\cmidrule(r){6-7}\cmidrule(r){8-9}
		
		$(n,\mathcal{L},\mathcal{A})$
		& Dim & & LGPR & ETPR & LGPR & ETPR & LGPR & ETPR \\
		\midrule
		
		$(2,3,2)$ & $33$ & max & 95 & 129 & 1.1 & 1.6 & 0\% & 0\% \\
		& & min & 66 & 100 & 0.7 & 1.2 & & \\
		& & \textbf{med} & \textbf{78.5} & \textbf{110.5}
		& \textbf{0.8} & \textbf{1.3} & & \\
		
		$(2,4,2)$ & $352$ & max & 384 & 1848 & 104.3 & 821.5 & 0\% & 0\% \\
		& & min & 184 & 273 & 54.3 & 140.6 & & \\
		& & \textbf{med} & \textbf{226.5} & \textbf{308.0}
		& \textbf{67.9} & \textbf{160.4} & & \\
		
		$(2,3,3)$ & $353$ & max & 578 & 1271 & 144.1 & 999.2 & 0\% & 0\% \\
		& & min & 193 & 259 & 60.0 & 142.0 & & \\
		& & \textbf{med} & \textbf{244.0} & \textbf{309.5}
		& \textbf{74.1} & \textbf{162.0} & & \\
		
		$(3,4,2)$ & $62$ & max & 2206 & 295 & 141.4 & 9.6 & 0\% & 0\% \\
		& & min & 86 & 139 & 2.4 & 4.1 & & \\
		& & \textbf{med} & \textbf{109.5} & \textbf{166.0}
		& \textbf{2.9} & \textbf{4.9} & & \\
		
		$(3,6,2)$ & $386$ & max & 745 & -- & 1139.1 & -- & 0\% & 10\% \\
		& & min & 218 & 307 & 263.3 & 760.6 & & \\
		& & \textbf{med} & \textbf{267.0} & \textbf{423.5}
		& \textbf{306.9} & \textbf{1076.7} & & \\
		
		$(3,7,2)$ & $705$ & max & -- & -- & -- & -- & 5\% & 100\% \\
		& & min & 338 & -- & 1808.5 & -- & & \\
		& & \textbf{med} & \textbf{401.5} & \textbf{--}
		& \textbf{2234.4} & \textbf{--} & & \\
		
		$(3,4,3)$ & $698$ & max & -- & -- & -- & -- & 5\% & 100\% \\
		& & min & 339 & -- & 1735.3 & -- & & \\
		& & \textbf{med} & \textbf{452.0} & \textbf{--}
		& \textbf{2327.7} & \textbf{--} & & \\
		\bottomrule
	\end{tabular*}
\end{table}
\section*{Data availability}
No data was used for the research described in the article.
\section*{Declaration of competing interest}
The authors declare that they have no financial or non-financial interests that are directly or indirectly related to the work reported in this paper.
\section*{Funding}
This work was partially supported by the Research Grants Council of the Hong Kong SAR Government (Grant No. 11215123) and the National Natural Science Foundation of China (Grant No. 72501218).

\begin{appendices}
	\numberwithin{equation}{section}
	\section{Proof of Lemma~\ref{nfprpure}}\label{prooflem2}
	This appendix provides the proof of Lemma~\ref{nfprpure}. We first introduce an auxiliary payoff notation. Given $\gamma\in\Lambda$, we define the expected payoff, leading by the sequence $\varpi^i\in W^i$, for player $i\in N$ as
	\begin{equation}\label{sqexpayoff}\textstyle
		g^i(\gamma;\varpi^i)=\sum\limits_{\tilde{\varpi}^i\in W^i,\varpi^i\subseteq\tilde{\varpi}^i}\gamma^i(\tilde{\varpi}^i)g^i(\tilde{\varpi}^i,\gamma^{-i}).
	\end{equation}
	Unlike the expected payoff $g^i(\gamma)$, which accounts for payoffs along all sequences, $g^i(\gamma;\varpi^i)$ includes only the payoffs along sequences that contain $\varpi^i$. Since \[\begin{array}{l}
		g^i_m(\varpi^i_{I^j_i}a,\gamma^{-i})\\
		\hspace{0.5cm}=\max\limits_{\tilde\gamma^i\in \Lambda^i} g^i(\tilde\gamma^i,\gamma^{-i};\varpi^i_{I^j_i}a)+\max\limits_{\tilde\gamma^i\in \Lambda^i,\tilde\gamma^i(\varpi^i_{I^j_i})=1}\sum\limits_{\tilde{\varpi}^i\in W^i,\tilde{\varpi}^i\cap A(I^j_i)=\emptyset}\gamma^i(\tilde{\varpi}^i)g^i(\tilde{\varpi}^i,\gamma^{-i}),
	\end{array}\]
	and the second term depends only on sequences that do not traverse information set $I^j_i$, Equation~(\ref{defbest}) in Definition~\ref{bestresponse} can be equivalently replaced by
	\[\textstyle\max\limits_{\tilde\gamma^i\in \Lambda^i} g^i(\tilde\gamma^i,\gamma^{-i};\varpi^i_{I^j_i}a)\geq\max\limits_{\tilde\gamma^i\in \Lambda^i} g^i(\tilde\gamma^i,\gamma^{-i};\varpi^i_{I^{j}_i}a').\]
	
	$\bm{(1)\Rightarrow(2)}$: Assume that (1) holds while (2) fails. Then there exists a sequence $\varpi^i_{I^j_i}a\nsubseteq\varpi^i$, with $j\in M_i,a\in A(I^j_i)$, such that $s^i(\varpi^i_{I^j_i}a)=1$ but $\varpi^i_{I^j_i}a$ is not a $I^j_i$-best-response sequence to $\gamma$. There exists $a'\in A(I^{j}_i)$ such that
	\[\textstyle\max\limits_{\tilde\gamma^i\in \Lambda^i} g^i(\tilde\gamma^i,\gamma^{-i};\varpi^i_{I^{j}_i}a)<\max\limits_{\tilde\gamma^i\in \Lambda^i} g^i(\tilde\gamma^i,\gamma^{-i};\varpi^i_{I^{j}_i}a'),\]
	producing a pure strategy $\tilde s^i$ such that $\gamma^i(\tilde s^i)\in\arg\max_{\tilde\gamma^i\in \Lambda^i} g^i(\tilde\gamma^i,\gamma^{-i};\varpi^i_{I^{j}_i}a')$ and $\gamma^i(\tilde s^i;\tilde\varpi^i)=\gamma^i(s^i;\tilde\varpi^i)$ for any $\tilde\varpi^i\in W^i$ with $a,a'\notin \tilde\varpi^i$. It follows that $g^i(\gamma^i(s^i),\gamma^{-i})<g^i(\gamma^i(\tilde s^i),\gamma^{-i})$ and hence $u^i(s^i,\sigma^{-i})< u^i(\tilde s^i,\sigma^{-i})$, leading to a contradiction.
	
	$\bm{(2)\Rightarrow(1)}$:
	Assume that condition (2) is satisfied. For any $j\in M_i,a\in A(I^j_i)$ with $s^i(\varpi^i_{I^{j}_i}a)=1$, it holds that
	\[\begin{array}{l}
		\max\limits_{\tilde\gamma^i\in \Lambda^i,\tilde\gamma^i(\varpi^i)=1} g^i(\tilde\gamma^i,\gamma^{-i};\varpi^i_{I^{j}_i}a)\\
		\hspace{0.2cm}\!=\!\sum\limits_{j_q\in M_i(\varpi^i_{I^{j}_i}a)}\!\max\limits_{\tilde\gamma^i\in \Lambda^i,\tilde\gamma^i(\varpi^i)=1}\!\sum\limits_{a_q\in A(I^{j_q}_i)}\! \tilde\gamma^i(\varpi^i_{I^{j_q}_i}a_q)g^i(\tilde\gamma^i,\gamma^{-i};\varpi^i_{I^{j_q}_i}a_q)\!+\!g^i(\varpi^i_{I^{j}_i}a,\gamma^{-i})\\
		\hspace{0.2cm}\!=\!\sum\limits_{j_q\in M_i(\varpi^i_{I^{j}_i}a)}\!\sum\limits_{a_q\in A(I^{j_q}_i)}s^i(\varpi^i_{I^{j_q}_i}a_q)\!\max\limits_{\tilde\gamma^i\in \Lambda^i,\tilde\gamma^i(\varpi^i)=1}\! g^i(\tilde\gamma^i,\gamma^{-i};\varpi^i_{I^{j_q}_i}a_q)\!+\!g^i(\varpi^i_{I^{j}_i}a,\gamma^{-i})\\
	\end{array}\]
	The second equality follows immediately from condition (2). Applying forward induction, one derives
	\[\begin{array}{l}
		\max\limits_{\tilde\gamma^i\in \Lambda^i,\tilde\gamma^i(\varpi^i)=1} g^i(\tilde\gamma^i,\gamma^{-i})\\
		\hspace{1cm}=\sum\limits_{j\in M_i(\varpi^i_\emptyset)}\max\limits_{\tilde\gamma^i\in \Lambda^i,\tilde\gamma^i(\varpi^i)=1}\sum\limits_{a\in A(I^j_i)} \tilde\gamma^i(\varpi^i_{I^j_i}a)g^i(\tilde\gamma^i,\gamma^{-i};\varpi^i_{I^j_i}a)+g^i(\varpi^i_\emptyset,\gamma^{-i})\\
		\hspace{1cm}=\sum\limits_{j\in M_i(\varpi^i_\emptyset)}\sum\limits_{a\in A(I^j_i)}s^i(\varpi^i_{I^j_i}a)\max\limits_{\tilde\gamma^i\in \Lambda^i,\tilde\gamma^i(\varpi^i)=1} g^i(\tilde\gamma^i,\gamma^{-i};\varpi^i_{I^j_i}a)+g^i(\varpi^i_\emptyset,\gamma^{-i})\\
		\hspace{1cm}=\sum\limits_{\tilde\varpi^i\in W^i}s^i(\tilde\varpi^i) g^i(\tilde\varpi^i,\gamma^{-i})\\
		\hspace{1cm}=g^i(\gamma^i(s^i),\gamma^{-i}).
	\end{array}\]
	It follows that $u^i(s^i,\sigma^{-i})=\max_{\tilde \sigma^i\in \Xi^i,\tilde \sigma^i(\varpi^i)=1} u^i(\tilde \sigma^i,\sigma^{-i})$, which completes the proof.
	
	\section{Sequence-Form Characterizations of Other Equilibrium Refinements}\label{formulations}
	In addition to the sequence-form characterization developed for normal-form proper equilibrium, related ideas can also be extended to other equilibrium concepts. This appendix briefly discusses how the proposed framework can be adapted to characterize several additional equilibrium concepts in the sequence form, including Nash equilibrium, normal-form perfect equilibrium, quasi-perfect equilibrium and quasi-proper equilibrium.
	
	The local best-response conditions in Definition~\ref{bestresponse} admits an equivalent global criterion over \(W^i\), as stated below.
	
	\begin{lemma}\label{propbrs}
		For player $i\in N$, $\varpi^i\in W^i$ is a best-response sequence to a given $\gamma\in\Lambda$ if and only if $g^i_m(\varpi^i,\gamma^{-i})\geq g^i_m(\tilde \varpi^i,\gamma^{-i})$ for any $\tilde \varpi^i\in W^i$.
	\end{lemma}
	\begin{proof}
		The sufficiency direction is immediate from the definition. We now show necessity. For any $i\in N,j\in M_i,a\in A(I^j_i)$, suppose that $\varpi^i_{I^j_i}a$ is an $I^j_i$-best-response sequence to $\gamma$, then $g^i_m(\varpi^i_{I^j_i}a,\gamma^{-i})= g^i_m(\varpi^i_{I^j_i},\gamma^{-i})$. Next, consider any $\varpi^i\in W^i$ that is a best-response sequence to a given $\gamma\in\Lambda$. By recursively invoking the local equality above along the sequence, one arrives at $g^i_m(\varpi^i,\gamma^{-i})= g^i_m(\varpi^i_\emptyset,\gamma^{-i})$. This identity immediately yields $g^i_m(\varpi^i,\gamma^{-i})\geq g^i_m(\tilde \varpi^i,\gamma^{-i})$ for any $\tilde \varpi^i\in W^i$, completing the proof.
	\end{proof}
	Accordingly, we introduce the definition of sequence-form Nash equilibrium.
	
	\begin{definition}\label{nfpr-def-sfne}
		Let $\Gamma$ be an extensive form game. A realization plan profile $\gamma^*\in \Lambda$ is defined as a sequence-form Nash equilibrium of $\Gamma$ if, for any player $i\in N$, $\gamma^{*i}(\varpi^i)=0$ whenever $g^i_m(\varpi^i,\gamma^{*-i})<g^i_m(\tilde \varpi^i,\gamma^{*-i})$ for some $\tilde \varpi^i\in W^i$.
	\end{definition}
	
	Next theorem established the equivalence between normal-form Nash equilibrium and sequence-form Nash equilibrium.
	
	\begin{theorem}\label{nfpr-thm-sfc1}
		For $(\sigma^*,\gamma^*)\in T$, $\sigma^*$ is a Nash equilibrium if and only if, for any player $i\in N$, $\gamma^{*i}(\varpi^i)=0$ whenever $g^i_m(\varpi^i,\gamma^{*-i})<g^i_m(\tilde \varpi^i,\gamma^{*-i})$ for some $\tilde \varpi^i\in W^i$.
	\end{theorem}\vspace{-12pt}
	\begin{proof}
		Assume that $\sigma^*$ is a Nash equilibrium. By Lemma~\ref{propbrs}, we show that any sequence $\varpi^i\in W^i$ satisfying $\gamma^{*i}(\varpi^i)>0$ is a best-response sequence to $\gamma^*$. Since $\gamma^{*i}(\varpi^i)=\sum_{s^i\in S^i}s^i(\varpi^i)\sigma^{*i}(s^i)>0$, there exists $s^i\in S^i$ such that $s^i(\varpi^i)=1$ and $\sigma^{*i}(s^i)>0$. Consequently, $u^i(s^i,\sigma^{*-i})\geq u^i(\tilde s^i,\sigma^{*-i})$ for all $\tilde s^i\in S^i$. It then follows from Lemma~\ref{nfprpure} that $\varpi^i$ is a best-response sequence to $\gamma^*$.
		
		Suppose that every $\varpi^i\in W^i$ with $\gamma^{*i}(\varpi^i)>0$ is a best-response sequence to $\gamma^*$. Consider $s^i\in S^i$ such that $\sigma^{*i}(s^i)>0$. Then, for any $\varpi^i$ with $s^i(\varpi^i) = 1$, we have $\gamma^{*i}(\varpi^i)=\sum_{\tilde s^i\in S^i}\tilde s^i(\varpi^i)\sigma^{*i}(\tilde s^i)\geq s^i(\varpi^i)\sigma^{*i}(s^i)>0$. Hence, $\varpi^i$ is a best-response sequence to $\gamma^*$. By Lemma~\ref{nfprpure}, this implies that $u^i(s^i, \sigma^{*-i}) \ge u^i(\tilde{s}^i, \sigma^{*-i})$ for all $\tilde{s}^i \in S^i$, and therefore $\sigma^*$ is a Nash equilibrium.
	\end{proof}
	
	We next consider normal-form perfect equilibrium. The following definition and theorem establish the corresponding sequence-form characterization.
	\begin{definition}\label{nfpr-def-sfpe}
		Let $\Gamma$ be an extensive form game. For any sufficiently small $\varepsilon > 0$, a totally realization plan profile $\gamma(\varepsilon)\in \Lambda$ is an $\varepsilon$-sequence-form perfect equilibrium of $\Gamma$ if $\gamma^i(\varepsilon;\varpi^i)\leq\varepsilon$ whenever $g^i_m(\varpi^i,\gamma^{-i}(\varepsilon)) < g^i_m(\tilde\varpi^i,\gamma^{-i}(\varepsilon))$ for all $i\in N$ and $\varpi^i,\tilde\varpi^i\in W^i$. A realization plan profile $\gamma^*\in \Lambda$ is defined as a sequence-form perfect equilibrium of game $\Gamma$ if $\gamma^*$ is a limit point of some sequence $\{\gamma(\varepsilon^k)\}_{k=1}^\infty$, where $\lim_{k\to\infty}\varepsilon^k=0$ and each $\gamma(\varepsilon^k)$ is an $\varepsilon^k$-sequence-form perfect equilibrium of $\Gamma$.
	\end{definition}
	\begin{theorem}\label{nfpr-lem-sfc3}
		For $(\sigma^*,\gamma^*)\in T$, $\sigma^*$ is a normal-form perfect equilibrium if and only if $\gamma^*$ is a sequence-form perfect equilibrium.
	\end{theorem}

	\begin{proof}
		We prove both directions separately.
		
		$(\Rightarrow)$ Assume that $\gamma^*$ is a sequence-form perfect equilibrium. By definition~\ref{nfpr-def-sfpe}, there exists a convergent sequence $\{\gamma(\varepsilon^k)\}^\infty_{k=1}$ of $\varepsilon^k$-sequence-form perfect equilibria such that $\lim_{k\to\infty} \varepsilon^k = 0$ and $\lim_{k\to\infty}\gamma(\varepsilon^k)=\gamma^*$. For each $k$, select a totally mixed strategy profile $\sigma^k$ satisfying $(\sigma^k,\gamma(\varepsilon^k))\in T$ and $\lim_{k\to\infty}\sigma^k=\sigma^*$. We claim that $\sigma^k$ constitutes an $\varepsilon^k$-normal-form perfect equilibrium.
		
		Fix $i\in N$ and $s^i,\tilde s^i\in S^i$ for which $u^i(s^i, \sigma^{-i}(\gamma(\varepsilon^k))) < u^i(\tilde s^{i},\sigma^{-i}(\gamma(\varepsilon^k)))$. By Lemma~\ref{nfprpure}, there exists a sequence $\varpi^i\in W^i$ with $s^i(\varpi^i)=1$ such that $g^i_m(\varpi^i,\gamma^{-i}(\varepsilon^k)) < g^i_m(\tilde\varpi^i,\gamma^{-i}(\varepsilon^k))$ for some $\tilde\varpi^i\in W^i$. Since $\gamma(\varepsilon^k)$ is an $\varepsilon^k$-sequence-form perfect equilibrium, it follows that $\gamma^{i}(\varepsilon^k;\varpi^i)\leq\varepsilon^k$. As $\gamma^{i}(\varepsilon^k;\varpi^i)=\sum_{s^i_q\in S^i}s^i_q(\varpi^i)\sigma^{ki}(s^i_q)$ and $s^i(\varpi^i)=1$, we have $\sigma^{ki}(s^i)\leq\gamma^{i}(\varepsilon^k;\varpi^i)\leq\varepsilon^k$. Thus, $\sigma^k$ is an $\varepsilon^k$-normal-form perfect equilibrium for each $k$. Therefore, $\sigma^*$ is a normal-form perfect equilibrium.
		
		$(\Leftarrow)$ Conversely, assume that $\sigma^*$ is a normal-form perfect equilibrium of $\Gamma$. By definition, there exists a sequence of totally mixed strategies $\{\sigma(\varepsilon^k)\}_{k=1}^{\infty}$ such that $\lim_{k\to\infty}\varepsilon^k=0$ and $\lim_{k\to\infty}\sigma(\varepsilon^k)=\sigma^*$, where each $\sigma(\varepsilon^k)$ is an $\varepsilon^k$-normal form perfect equilibrium. Consider a specific $\sigma(\varepsilon^k)$ with sufficiently large $k$. If $u^i(s^i, \sigma^{-i}(\varepsilon^k)) < u^i(\tilde s^{i},\sigma^{-i}(\varepsilon^k))$ holds for some $\tilde s^i\in S^i$, then $\sigma^i(\varepsilon^k;s^i)\leq\varepsilon^k$. We now show that $\gamma(\sigma(\varepsilon^k))$ is a $\sqrt{\varepsilon^k}$-sequence-form perfect equilibrium for all $k$. 
		
		Take $i\in N,\varpi^i,\tilde\varpi^i\in W^i$ with $g^i_{m}(\varpi^i,\gamma^{-i}(\sigma(\varepsilon^k)))\!<\! g^i_{m}(\tilde\varpi^i,\gamma^{-i}(\sigma(\varepsilon^k)))$. By Lemma~\ref{nfprpure}, for any $s^{i}\in S^i$ with $s^{i}(\varpi^i)=1$, there exists $\tilde s^i\in S^i$ such that $u^i(s^i, \sigma^{-i}(\varepsilon^k)) < u^i(\tilde s^{i},\sigma^{-i}(\varepsilon^k))$. Hence, for all such $s^{i}\in S^i$, we have $\sigma^i(\varepsilon^k;s^{i})\leq\varepsilon^k$. It follows that $\gamma^{i}(\sigma(\varepsilon^k);\varpi^i)=\sum_{s^i\in S^i}s^i(\varpi^i)\sigma^i(\varepsilon^k;s^i)\leq\sqrt{\varepsilon^k}$, which establishes the proof.
	\end{proof}
	
	The same perspective can also be applied to quasi-perfect equilibrium and quasi-proper equilibrium. The required modifications are largely confined to the corresponding local comparison conditions. For brevity, we record the definitions below and omit the corresponding equivalence statements and proofs.
	
	\begin{definition}\label{nfpr-def-sfqpe}
		Let $\Gamma$ be an extensive form game. For any sufficiently small $\varepsilon > 0$, a totally realization plan profile $\gamma(\varepsilon)\in \Lambda$ is an $\varepsilon$-sequence-form quasi-perfect equilibrium of $\Gamma$ if $\gamma^i(\varepsilon;\varpi^i_{I^j_i}a)\leq\varepsilon$ whenever $g^i_m(\varpi^i_{I^j_i}a,\gamma^{-i}(\varepsilon)) < g^i_m(\varpi^i_{I^j_i}a',\gamma^{-i}(\varepsilon))$ for all $i\in N,j\in M_i$, and $a,a'\in A(I^j_i)$. A realization plan profile $\gamma^*\in \Lambda$ is defined as a sequence-form quasi-perfect equilibrium of game $\Gamma$ if $\gamma^*$ is a limit point of some sequence $\{\gamma(\varepsilon^k)\}_{k=1}^\infty$, where $\lim_{k\to\infty}\varepsilon^k=0$ and each $\gamma(\varepsilon^k)$ is an $\varepsilon^k$-sequence-form quasi-perfect equilibrium of $\Gamma$.
	\end{definition}
	\begin{definition}\label{nfpr-def-sfqpr}
		Let $\Gamma$ be an extensive form game. For any sufficiently small $\varepsilon > 0$, a totally realization plan profile $\gamma(\varepsilon)\in \Lambda$ is an $\varepsilon$-sequence-form quasi-proper equilibrium of $\Gamma$ if $\gamma^i(\varepsilon;\varpi^i_{I^j_i}a)\leq\varepsilon\gamma^i(\varepsilon;\varpi^i_{I^j_i}a')$ whenever $g^i_m(\varpi^i_{I^j_i}a,\gamma^{-i}(\varepsilon)) < g^i_m(\varpi^i_{I^j_i}a',\gamma^{-i}(\varepsilon))$ for all $i\in N,j\in M_i$, and $a,a'\in A(I^j_i)$. A realization plan profile $\gamma^*\in \Lambda$ is defined as a sequence-form quasi-proper equilibrium of game $\Gamma$ if $\gamma^*$ is a limit point of some sequence $\{\gamma(\varepsilon^k)\}_{k=1}^\infty$, where $\lim_{k\to\infty}\varepsilon^k=0$ and each $\gamma(\varepsilon^k)$ is an $\varepsilon^k$-sequence-form quasi-proper equilibrium of $\Gamma$.
	\end{definition}
	
	\section{Proof of Theorem~\ref{nfpr-the-sfc1}}\label{proofthm3}
	
	This appendix is devoted to the proof of Theorem~\ref{nfpr-the-sfc1}. We begin by introducing some auxiliary notation and establishing several preliminary results needed for the proof. For $i\in N$ and  $\pi^i=\{\varpi^i_1,\varpi^i_2,\dots,\varpi^i_{\kappa_i(\varpi^i_\emptyset)}\}\in \Pi^i$, we recursively construct the subsets \(\pi_l^i\) for \(l\in K^-_i\) by setting \(\pi_0^i=\emptyset\) and
	\begin{equation}\label{rssc}\textstyle
		\pi_l^i
		= \{\varpi^i\in \pi_{l-1}^i\mid \varpi_l^i\nsubseteq \varpi^i\}\cup\{\varpi_l^i\}.
	\end{equation}
	That is, \(\pi_l^i\) is obtained by adding \(\varpi_l^i\) to \(\pi_{l-1}^i\) and removing all sequences in \(\pi_{l-1}^i\) that are extensions of \(\varpi_l^i\). By Definition~\ref{seqcons}, we have \(\pi_l^i\in\mathcal{E}^i\) for every $l\in K^-_i$.
	
	For $i\in N$ and $\gamma^i\in\Lambda^i_{++}$, take the ordering $\pi^i(\gamma^i)=\{\varpi^i_1,\varpi^i_2,\dots,\varpi^i_{\kappa_i(\varpi^i_\emptyset)}\}$, from which the subsets $\pi^i_l(\gamma^i),\,l\in K^-_i,$ are recursively generated according to (\ref{rssc}). Associated with this ordering, let $S^i(\pi^i(\gamma^i))=\{s^i_1,s^i_2,\dots,s^i_{\kappa_i(\varpi^i_\emptyset)}\}$. We then define $\tilde\sigma^i(\gamma^i)=(\tilde\sigma^i(\gamma^i;s^i):s^i\in S^i)$ by assigning $\tilde\sigma^i(\gamma^i;s^i)=0$ for all $s^i\in S^i\setminus S^i(\pi^i(\gamma^i))$, and specifying $\tilde\sigma^i(\gamma^i;s^i)$ for $s^i\in S^i(\pi^i(\gamma^i))$ through (\ref{ms2rp}). In particular, for $l\in K_i$, we set
	\[\textstyle
	\tilde\sigma^i(\gamma^i;s^i_l)=\sum\limits_{\varpi^i\in \pi^i_l(\gamma^i)}\gamma^i(\varpi^i)-\sum\limits_{\varpi^i\in \pi^i_{l-1}(\gamma^i)}\gamma^i(\varpi^i).
	\]
	It follows that $\sum_{s^i\in S^i}\tilde\sigma^i(\gamma^i;s^i)=1$ and $u^i(\tilde\sigma(\gamma))=g^i(\gamma)$, where $\tilde\sigma(\gamma)=(\tilde\sigma^i(\gamma^i):i\in N)$. Moreover, whenever $\tilde\sigma^i(\gamma^i;s^i)\geq0$ for all $i\in N,s^i\in S^i$, the profile $\tilde\sigma(\gamma)$ is feasible in $\Xi$, and hence $(\tilde\sigma(\gamma),\gamma)\in T$.
	\begin{lemma}\label{snpc}
		For any $\gamma^i\in\Lambda^i$, $\gamma^i\in\Lambda^i(\delta(\varepsilon))$ holds if and only if \begin{equation}\label{scs}\textstyle
			\sum\limits_{\varpi^i\in \pi^i_l(\gamma^i)}\gamma^i(\varpi^i)\geq\sum\limits_{q=1}^{l}r_q^i(\delta(\varepsilon)),\,l\in K^-_i.
		\end{equation}
	\end{lemma}
	\begin{proof}
		\((\Rightarrow)\) The forward implication is immediate, since the constraints in \eqref{scs} are encompassed by those in \eqref{dflamda}.
		
		\((\Leftarrow)\) Note that, for any \(E^i\in\mathcal{E}^i\), $
		\sum_{\varpi^i\in E^i}\gamma^i(\varpi^i)
		\geq
		\sum_{\varpi^i\in \pi^i_{q^i(E^i)}(\gamma^i)}
		\gamma^i(\varpi^i),
		$ which follows directly from the definition of \(\pi^i_{q^i(E^i)}(\gamma^i)\). Consequently, by \eqref{scs}, the constraints in \eqref{dflamda} are satisfied, implying that \(\gamma^i\in\Lambda^i(\delta(\varepsilon))\).
	\end{proof}
	\begin{lemma}\label{nfper-lem-sfc4}
		Let $\gamma^*(\varepsilon)\in\Lambda(\delta(\varepsilon))$ and $\pi^i(\gamma^{*i}(\varepsilon))=\{\varpi^i_1,\varpi^i_2,\dots,\varpi^i_{\kappa_i(\varpi^i_\emptyset)}\}$ for $i\in N$. The profile $\gamma^*(\varepsilon)$ is a Nash equilibrium of $\Gamma_s(\delta(\varepsilon))$ if and only if the following two conditions hold.
		\begin{enumerate}
			\item Any sequence $\varpi^i_{I^j_i}a$, $i\in N,j\in M_i,a\in I^j_i$, with $\varpi^i_{I^j_i}a\in W^i\setminus \pi^i(\gamma^{*i}(\varepsilon))$ is an $I^j_i$-best-response sequence to $\gamma^*(\varepsilon)$.
			\item For every player $i\in N$, the profile satisfies
			\begin{equation}\label{extremeequation}\textstyle
				\sum\limits_{\varpi^i\in \pi^i_l(\gamma^{*i}(\varepsilon))}\gamma^{*i}(\varepsilon;\varpi^i)=\sum\limits_{q=1}^{l}r_q^i(\delta(\varepsilon)),
			\end{equation}
			whenever $g^i_m(\varpi^i_l,\gamma^{*-i}(\varepsilon)) < g^i_m(\varpi^i_{l+1},\gamma^{*-i}(\varepsilon)),\; l\in K^-_i$.
		\end{enumerate}
	\end{lemma}\vspace{-12pt}
	\begin{proof}
		$(\Rightarrow)$ Assume that the profile $\gamma^*(\varepsilon)$ is a Nash equilibrium of the game $\Gamma_s(\delta(\varepsilon))$. Suppose, for contradiction, that the first condition fails to hold. Then there exists an information set $I_i^j$ and a sequence $\varpi^i_{I_i^j} a$ that is not an $I_i^j$-best-response sequence to $\gamma^*(\varepsilon)$, while another sequence $\varpi^i_{I_i^j} a'$ is an $I_i^j$-best-response sequence to $\gamma^*(\varepsilon)$. Moreover, for any sequence $\varpi^i_{I_i^{j_q}} a_q \in W^i \setminus \pi^i(\gamma^{*i}(\varepsilon))$ satisfying either $\varpi^i_{I_i^j} a \subset \varpi^i_{I_i^{j_q}} a_q$ or $\varpi^i_{I_i^j} a' \subset \varpi^i_{I_i^{j_q}} a_q$, it holds that $\varpi^i_{I_i^{j_q}} a_q$ is an $I_i^{j_q}$-best-response sequence to $\gamma^*(\varepsilon)$. Accordingly, one can construct two pure strategies, $s^i$ and $\tilde s^i$ with $s^i(\varpi^i_{I_i^j} a) = 1$ and $\tilde s^i(\varpi^i_{I_i^j} a') = 1$, such that, whenever $s^i(\varpi^i) = 1$ or $\tilde{s}^i(\varpi^i) = 1$, it holds that $\varpi^i \in W^i \setminus \pi^i(\gamma^i)$. It then follows that $g^i(\gamma^i(s^i), \gamma^{*-i}(\varepsilon)) < g^i(\gamma^i(\tilde{s}^i), \gamma^{*-i}(\varepsilon))$. Define $\tilde{\gamma}^i = \gamma^{*i}(\varepsilon) + \tilde \varepsilon (\gamma^i(\tilde{s}^i) - \gamma^i(s^i))$ with sufficiently small $\tilde \varepsilon>0$. We obtain from Lemma~\ref{snpc} that $\tilde{\gamma}^i \in \Lambda^i(\delta(\varepsilon))$ and $g^i(\tilde{\gamma}^i, \gamma^{*-i}(\varepsilon)) > g^i(\gamma^{*i}(\varepsilon), \gamma^{*-i}(\varepsilon))$, which contradicts the assumption that $\gamma^*(\varepsilon)$ is a Nash equilibrium. Hence, the first condition must hold.
		
		Now suppose that the second condition fails. Then there exists a player $i\in N$ and an index $l\in K^-_i$ with $g^i_m(\varpi^i_l,\gamma^{*-i}(\varepsilon)) < g^i_m(\varpi^i_{l+1},\gamma^{*-i}(\varepsilon))$ such that $\sum_{\varpi^i\in \pi^i_l(\gamma^{*i}(\varepsilon))}\gamma^{*i}(\varepsilon;\varpi^i)>\sum_{q=1}^{l}r_q^i(\delta(\varepsilon))$. Let $S^i(\pi^i(\gamma^{*i}(\varepsilon)))=\{s^i_1,s^i_2,\dots,s^i_{\kappa_i(\varpi^i_\emptyset)}\}$. Then $g^i(\gamma^i(s^i_l),\gamma^{*-i}(\varepsilon)) < g^i(\gamma^i(s^i_{l+1}),\gamma^{*-i}(\varepsilon))$. Define $\tilde{\gamma}^i = \gamma^{*i}(\varepsilon) + \tilde \varepsilon (\gamma^i(s^i_{l+1}) - \gamma^i(s^i_{l}))$ for sufficiently small $\tilde \varepsilon>0$. Lemma~\ref{snpc} implies that $\tilde{\gamma}^i \in \Lambda^i(\delta(\varepsilon))$ and $g^i(\tilde{\gamma}^i, \gamma^{*-i}(\varepsilon)) > g^i(\gamma^{*i}(\varepsilon), \gamma^{*-i}(\varepsilon))$, which contradicts the Nash equilibrium condition for $\gamma^*(\varepsilon)$. Therefore, the second condition must also hold.
		
		$(\Leftarrow)$ Assume that $\gamma^*(\varepsilon)$ satisfies the required two conditions, we now prove that it constitutes a Nash equilibrium of $\Gamma_s(\delta(\varepsilon))$. It suffices to verify that, for each player $i \in N$, the realization plan $\gamma^{*i}(\varepsilon)$ is a best response to $\gamma^{*-i}(\varepsilon)$ under the constraints $\gamma^i\in\Lambda^i$ and $
		\sum_{\varpi^i\in \pi^i_l(\gamma^{*i}(\varepsilon))}\gamma^i(\varpi^i)\geq\sum_{q=1}^{l}r_q^i(\delta(\varepsilon)),\;l\in K^-_i$.
		Let $S^i(\pi^i(\gamma^{*i}(\varepsilon)))=\{s^i_1,s^i_2,\dots,s^i_{\kappa_i(\varpi^i_\emptyset)}\}$. Since $\sum_{q=1}^l\tilde\sigma^i(\gamma^{i}(\varepsilon);s^i_q)=\sum_{\varpi^i\in \pi^i_l(\gamma^{*i}(\varepsilon))}\gamma^i(\varpi^i)$, it is equivalent to prove that $\tilde\sigma^i(\gamma^{*i}(\varepsilon))$ is a best response to $\tilde\sigma^{-i}(\gamma^{*-i}(\varepsilon))$ under the constraints
		\begin{equation}\label{mixpro}
			\begin{array}{l}
				\sum\limits_{q=1}^l\sigma^i(s^i_q)\geq\sum\limits_{q=1}^{l}r_q^i(\delta(\varepsilon)),\; l\in K^-_i,\\
				\sigma^i(s^i)=0,\; s^i\in S^i\setminus S^i(\pi^i(\gamma^{*i}(\varepsilon))),\\
				\sum\limits_{s^i\in S^i}\sigma^i(s^i)=1.
			\end{array}
		\end{equation}
		The first condition implies that $u^i(s^i_l,\tilde\sigma^{-i}(\gamma^{*-i}(\varepsilon)))=g^i_m(\varpi^i_l,\gamma^{*-i}(\varepsilon))$, and hence $u^i(s^i_l,\tilde\sigma^{-i}(\gamma^{*-i}(\varepsilon)))\leq u^i(s^i_{l+1},\tilde\sigma^{-i}(\gamma^{*-i}(\varepsilon)))$ for $l\in K^-_i$. Moreover, the second condition ensures that $\sum_{q=1}^l\tilde\sigma^i(\gamma^{*i}(\varepsilon);s^i_q)=\sum_{q=1}^{l}r_q^i(\delta(\varepsilon))$ whenever $u^i(s^i_l,\tilde\sigma^{-i}(\gamma^{*-i}(\varepsilon)))< u^i(s^i_{l+1},\tilde\sigma^{-i}(\gamma^{*-i}(\varepsilon)))$. Hence, $\tilde\sigma^i(\gamma^{*i}(\varepsilon))$ is a best response to $\tilde\sigma^{-i}(\gamma^{*-i}(\varepsilon))$ under the constraints~(\ref{mixpro}). This completes the proof.
	\end{proof}
	We next present the formal proof of Theorem~\ref{nfpr-the-sfc1}.
	
	$(\Rightarrow)$ If $\gamma^*(\varepsilon)$ is a Nash equilibrium of $\Gamma_s(\delta(\varepsilon))$, we establish that $\gamma^*(\varepsilon)$ constitutes a $2\varepsilon$-sequence-form proper equilibrium under Definition~\ref{nfpr-def-sfc1}. For each player $i\in N$, let $\pi^i(\gamma^{*i}(\varepsilon))=\{\varpi^i_1,\varpi^i_2,\dots,\varpi^i_{\kappa_i(\varpi^i_\emptyset)}\}$. It suffices to show that $\gamma^{*i}(\varepsilon; \varpi^i_l) \leq 2\varepsilon\,\gamma^{*i}(\varepsilon; {\varpi}^i_{l+1})$ whenever $g^i_m(\varpi^i_l,\gamma^{*-i}(\varepsilon)) \;<\; g^i_m(\varpi^i_{l+1},\gamma^{*-i}(\varepsilon))$. When this inequality holds, the equation~(\ref{extremeequation}) applies, implying that \[\textstyle\gamma^{*i}(\varepsilon; \varpi^i_l)\leq \sum\limits_{q=1}^{l} r_q^i(\delta(\varepsilon))=\prod\limits^{\kappa_i(\varpi^i_\emptyset)-l-1}_{w=1}\delta^i_w(\varepsilon)\sum\limits_{q=1}^{l}\prod\limits^{\kappa_i(\varpi^i_\emptyset)-q}_{w=\kappa_i(\varpi^i_\emptyset)-l}\delta^i_w(\varepsilon)\] and \[\textstyle\gamma^{*i}(\varepsilon; {\varpi}^i_{l+1})\geq r_{\kappa_i(\varpi^i_\emptyset)-l-1}^i(\delta(\varepsilon))=\prod\limits^{l+1}_{q=1}\delta^i_q(\varepsilon).\] Hence, it follows directly that $\gamma^{*i}(\varepsilon; \varpi^i_l) \leq 2\delta^i_{\kappa_i(\varpi^i_\emptyset)-l}(\varepsilon)\gamma^{*i}(\varepsilon; {\varpi}^i_{l+1})$, thereby completing the proof.
	
	$(\Leftarrow)$ Assume that $\gamma^*(\varepsilon)$ is an $\varepsilon$-sequence-form proper equilibrium. Our goal is to construct a perturbation vector $\delta(\tilde\varepsilon)$ such that $\gamma^*(\varepsilon)$ is a Nash equilibrium of the perturbed game $\Gamma_s(\delta(\tilde\varepsilon))$. By Definition~\ref{nfpr-def-sfc1}, for every $i\in N$ and any $\varpi^i,\tilde\varpi^i\in W^i$, it holds that $\gamma^{*i}(\varepsilon;\varpi^i)\leq \varepsilon \,\gamma^{*i}(\varepsilon;\tilde \varpi^i)$ whenever $g^i_m(\varpi^i,\gamma^{*-i}(\varepsilon)) < g^i_m(\tilde \varpi^i,\gamma^{*-i}(\varepsilon)).$ We write $\pi^i(\gamma^{*i}(\varepsilon))=\{\varpi^i_1,\dots,\varpi^i_{\kappa_i(\varpi^i_\emptyset)}\}$. Let $w_0=\max_{i\in N}\{|W^i|\}$ and define $\tilde\varepsilon=\varepsilon^{1/w_0}$. For each $i\in N$, we set $r_{0}^i(\delta(\tilde\varepsilon))=\varepsilon\min_{\varpi^i\in W^i}\{\gamma^{*i}(\varepsilon;\varpi^i)\}$ and define $r_l^i(\delta(\tilde\varepsilon))$, $l\in K^-_i$, recursively by
	\[\textstyle
	r_l^i(\delta(\tilde\varepsilon)) = 
	\begin{cases}
		\sum\limits_{\varpi^i\in \pi^i_l(\gamma^{*i}(\varepsilon))}\gamma^{*i}(\varepsilon;\varpi^i)-\sum\limits_{q=1}^{l-1} r_q^i(\delta(\tilde\varepsilon)), \\
		\hspace{4cm}\text{if } g^i_m(\varpi^i_l,\gamma^{*-i}(\varepsilon)) < g^i_m(\varpi^i_{l+1},\gamma^{*-i}(\varepsilon)), \\
		\tfrac{1}{\tilde\varepsilon}\, r_{l-1}^i(\delta(\tilde\varepsilon)),  \hspace{2cm}\text{otherwise}.
	\end{cases}
	\]
	The perturbation is then given by 
	\[
	\delta^i_l(\tilde\varepsilon)=\frac{r_{\kappa_i(\varpi^i_\emptyset)-l}^i(\delta(\tilde\varepsilon))}{r_{\kappa_i(\varpi^i_\emptyset)-l+1}^i(\delta(\tilde\varepsilon))},\; l\in K^-_i,
	\]
	which satisfies $0<\delta^i_l(\tilde\varepsilon)\leq \tilde\varepsilon$. Hence, the perturbation vector $\delta(\tilde\varepsilon)$ is well defined. By Lemma~\ref{nfper-lem-sfc4}, it follows that $\gamma^*(\varepsilon)$ is a Nash equilibrium of the perturbed game $\Gamma_s(\delta(\tilde\varepsilon))$. This completes the proof.
	
	\section{A Compact Dynamic-Programming Formulation}\label{compact}
	
	The main text defines the restricted sequence sets $\mathcal E^i$ and the perturbed realization-plan space $\Lambda^i(\delta(\varepsilon))$. We retain these definitions and construct a compact formulation of their constraints. Let $d_i=\kappa_i(\varpi^i_\emptyset)-1$ and $W_+^i=W^i\setminus\{\varpi^i_\emptyset\}$. For $q\geq0$, write
	$\mathcal E_q^i=\{E^i\in\mathcal E^i:q^i(E^i)=q\}$.
	For every attainable $q>0$, the constraints associated with $E^i\in\mathcal E_q^i$ are equivalent to
	\begin{equation}\label{eq:fc-weightwise-minimum}\textstyle
		\min\limits_{E^i\in\mathcal E_q^i}
		\sum\limits_{\varpi^i\in E^i}\gamma^i(\varpi^i)
		\geq\sum\limits_{\ell=1}^{q}r_\ell^i(\delta(\varepsilon)).
	\end{equation}
	To express \eqref{eq:fc-weightwise-minimum} through dynamic programming, we consider, for each sequence $\varpi^i$, the case in which $\varpi^i$ is omitted from a restricted sequence set. The corresponding selections consist of sequences extending $\varpi^i$ that satisfy the relevant restrictions, with the sum of $\kappa_i$ over the selected sequences equal to $q$. When $\varpi^i=\varpi_\emptyset^i$, these selections are precisely the sets in $\mathcal E_q^i$. For each feasible $q$, let $\varphi^i_{\varpi^i}(q)$ denote the expression used in the linear formulation for the minimum probability sum over these selections. An auxiliary variable is introduced unless the expression can be represented directly by zero or an existing variable. Following the dynamic-programming recursion, we construct the linear inequalities from sequences with no successor information sets toward their preceding sequences. These inequalities ensure that each auxiliary variable does not exceed the corresponding exact minimum, without requiring equality. The original perturbation bounds are then imposed on the final expressions $\varphi^i_{\varpi_\emptyset^i}(q)$.
	
	\paragraph{Initial case}
	If $M_i(\varpi^i)=\emptyset$, no sequence extending $\varpi^i$ is available for selection. Thus, only the empty selection is possible, with $q=0$ and a probability sum of zero. Accordingly, define
	\begin{equation}\label{eq:fc-D-definition}
		\varphi^i_{\varpi^i}(0)=0.
	\end{equation}
	No entry with $q>0$ is feasible in this case. Since $\kappa_i$ is positive for every nonempty sequence, the same zero expression applies to $\varphi^i_{\varpi^i}(0)$ for any $\varpi^i$.
	
	\paragraph{Recursion over action sequences at each information set.}
	Consider an information set $I_i^j$ for which $\varphi^i_{\varpi^i_{I_i^j}a}(q)$ has already been defined for every $a\in A(I_i^j)$ and feasible $q$. Fix an ordering $A(I_i^j)=\{a_1,\ldots,a_{m_j}\}$. Let $\varphi^i_{j,l,\theta}(q)$ denote the expression used to bound the minimum probability sum from below over selections associated with the first $l$ action sequences, including sequences extending any omitted action sequence, for which the sum of $\kappa_i$ equals $q$. Here $\theta=0$ indicates that all $l$ action sequences are included, whereas $\theta=1$ indicates that at least one is omitted. Initially, set $\varphi^i_{j,0,0}(0)=0$; all other entries with $l=0$ are infeasible.
	
	For $l=1,\ldots,m_j$, suppose the first $l-1$ action sequences correspond to an entry $\varphi^i_{j,l-1,\theta}(q')$. If $\varpi^i_{I_i^j}a_l$ is included, no sequence extending it can belong to the same restricted sequence set. Its inclusion adds $\gamma^i(\varpi^i_{I_i^j}a_l)$ to the probability sum and $\kappa_i(\varpi^i_{I_i^j}a_l)$ to $q'$, while leaving $\theta$ unchanged. Accordingly, impose
	\begin{equation}\label{eq:Dstates}
		\varphi^i_{j,l,\theta}
		\bigl(q'+\kappa_i(\varpi^i_{I_i^j}a_l)\bigr)
		\leq
		\varphi^i_{j,l-1,\theta}(q')
		+\gamma^i(\varpi^i_{I_i^j}a_l).
	\end{equation}
	
	If $\varpi^i_{I_i^j}a_l$ is omitted, sequences extending it may still be selected. For each feasible $q''\geq0$, their minimum probability sum is bounded below by $\varphi^i_{\varpi^i_{I_i^j}a_l}(q'')$. Since the current action sequence is omitted, the resulting omission index is one, regardless of $\theta$. Thus, impose
	\begin{equation}\label{eq:fc-action-recursion}
		\varphi^i_{j,l,1}(q'+q'')
		\leq
		\varphi^i_{j,l-1,\theta}(q')
		+\varphi^i_{\varpi^i_{I_i^j}a_l}(q'').
	\end{equation}	
	These inequalities are imposed for all feasible combinations of their indices. Any feasible entry with $q=0$ is represented directly by zero. After all $m_j$ action sequences have been considered, only entries with $\theta=1$ are retained, ensuring that at least one action sequence at $I_i^j$ is omitted. Entries with $\theta=0$ are retained at earlier stages because a subsequent action sequence may still be omitted.
	
	\paragraph{Recursion over successor information sets.}
	Consider a sequence $\varpi^i$ for which $\varphi^i_{j,m_j,1}(q)$ has been constructed for every $j\in M_i(\varpi^i)$ and feasible $q$. Fix an ordering $M_i(\varpi^i)=\{j_1,\ldots,j_s\}$, where $s\geq1$. Since selections associated with distinct successor information sets involve disjoint sequences, their sums of $\kappa_i$ and their probability sums are additive.
	
	Let $\varphi^i_{\varpi^i,r}(q)$ denote the expression associated with the minimum probability sum over selections from the first $r$ information sets for which the sum of $\kappa_i$ equals $q$. Define
	\begin{equation}\label{eq:Hstate}
		\varphi^i_{\varpi^i,1}(q)
		=\varphi^i_{j_1,m_{j_1},1}(q),
		\qquad
		\varphi^i_{\varpi^i}(q)
		=\varphi^i_{\varpi^i,s}(q).
	\end{equation}
	For $r=2,\ldots,s$, impose
	\begin{equation}\label{eq:fc-skip-recursion}
		\varphi^i_{\varpi^i,r}(q'+q'')
		\leq
		\varphi^i_{\varpi^i,r-1}(q')
		+\varphi^i_{j_r,m_{j_r},1}(q'')
	\end{equation}
	for all feasible nonnegative integers $q'$ and $q''$. No additional omission index is required, since each completed information-set expression already accounts for the omission of at least one action sequence at its information set.
	
	This completes the construction of $\varphi^i_{\varpi^i}(q)$ for all feasible $q$ when $\varpi^i$ is omitted. Applying the recursion successively to preceding sequences eventually yields $\varphi^i_{\varpi^i_\emptyset}(q)$, corresponding to the minimum probability sum over $\mathcal E_q^i$.
	
	\paragraph{Construction of the extended linear system.}
	We retain only the expressions needed to obtain $\varphi^i_{\varpi^i_\emptyset}(q)$ for attainable $q>0$, together with all inequalities required for their recursive construction. Expressions that can be represented by zero, existing realization-plan variables, or previously defined expressions require no additional auxiliary variables.
	
	For each attainable $q>0$, impose the original perturbation lower bound
	\begin{equation}\label{eq:general-terminal}\textstyle
		\varphi^i_{\varpi^i_\emptyset}(q)
		\geq
		\sum\limits_{\ell=1}^{q}r_\ell^i(\delta(\varepsilon)).
	\end{equation}
	
	Let $\varphi^i$ collect the auxiliary variables introduced for player $i$, and define $\mathbf z^i=(\gamma^i,\varphi^i)$. Collecting these variables for all players into $\mathbf z=(\gamma,\varphi)$, the recursive inequalities, the perturbation bounds in \eqref{eq:general-terminal}, and the original realization-plan nonnegativity and flow constraints yield
	\begin{equation}\label{eq:general-matrix}
		\mathbf A\mathbf z\geq\mathbf b(\delta(\varepsilon)),
		\quad
		\mathbf G\mathbf z=\mathbf c.
	\end{equation}
	Here $\mathbf G=[\mathbf G_\gamma\ \ 0]$ represents the original flow equations, including $\gamma^i(\varpi^i_\emptyset)=1$ for each player $i$, and $\mathbf c$ denotes their right-hand side.
	
	The resulting linear system has polynomial size. For each player $i$, the expressions $\varphi^i_{\varpi^i}(q)$, $\varphi^i_{j,h,\theta}(q)$, and $\varphi^i_{\varpi^i,t}(q)$ have at most $|W_+^i|+1$, $2|W_+^i|$, and $|M_i|$ distinct combinations of indices excluding $q$, respectively. Since $0\leq q\leq d_i$, the number of auxiliary variables is bounded by
	\[
	O\bigl((|W_+^i|+|M_i|+1)(d_i+1)\bigr).
	\]
	Furthermore, each recursive combination considers at most $(d_i+1)(d_i+2)/2$ pairs of nonnegative integers $q'$ and $q''$ satisfying $q'+q''\leq d_i$. Since the recursion involves $|W_+^i|$ action sequences and at most $|M_i|$ successive combinations of information sets, with only two possible values of the omission index, the number of recursive inequalities is bounded by
	\[
	O\bigl((|W_+^i|+|M_i|+1)(d_i+1)^2\bigr).
	\]
	By the definition of $\kappa_i$, $d_i=\sum_{j\in M_i}(m_j-1)\leq |W_+^i|$. Hence, both bounds are polynomial in the size of the explicit sequence-form representation.
	
	\begin{theorem}\label{thm:general-projection}
		For every perturbation vector $\delta(\varepsilon)$,
		\begin{equation}\label{eq:fc-projection}\textstyle
			\operatorname{proj}_{\gamma}
			\left\{
			\mathbf z:
			\mathbf A\mathbf z\geq\mathbf b(\delta(\varepsilon)),
			\;
			\mathbf G\mathbf z=\mathbf c
			\right\}
			=
			\prod\limits_{i\in N}\Lambda^i(\delta(\varepsilon)).
		\end{equation}
	\end{theorem}
	\begin{proof}
		Fix a player $i$ and a realization plan $\gamma^i$. For the purpose of the proof, define $\bar\varphi^i$ recursively by taking the minimum over all feasible combinations used in the preceding construction, with the initial expressions evaluated at $\gamma^i$. These values are distinct from the auxiliary variables in the linear system. The inclusion and omission cases cover all possible selections of each action sequence, while retaining only the final expressions with $\theta=1$ enforces the restriction at each information set. Moreover, selections associated with distinct successor information sets involve disjoint sequences, allowing their sums of $\kappa_i$ and probability sums to be combined independently. Thus, every selection satisfying the relevant restrictions is represented by the recursion, and every recursive combination preserves these restrictions. Induction in the prescribed order yields, for every attainable $q$,
		\begin{equation}\label{eq:fc-root-minimum}\textstyle
			\bar\varphi^i_{\varpi_\emptyset^i}(q;\gamma^i)
			=
			\min_{E^i\in\mathcal E_q^i}
			\sum\limits_{\varpi^i\in E^i}\gamma^i(\varpi^i).
		\end{equation}
		
		Suppose that $\mathbf z=(\gamma,\varphi)$ satisfies \eqref{eq:general-matrix}. We show by induction that each $\varphi$ expression is no greater than its corresponding exact value. The claim holds for the initial expressions, which are represented directly by zero or realization-plan variables. For each subsequent expression, the recursive inequalities bound it above by every feasible sum of preceding expressions. By the induction hypothesis, these sums do not exceed their corresponding exact values. Hence, the resulting expression cannot exceed the minimum of the exact sums. Consequently, for every attainable $q>0$ and $E^i\in\mathcal E_q^i$,
		\[\textstyle
			\sum\limits_{\ell=1}^{q}r_\ell^i(\delta(\varepsilon))
			\leq \varphi^i_{\varpi_\emptyset^i}(q)
			\leq \bar\varphi^i_{\varpi_\emptyset^i}(q;\gamma^i)
			\leq \sum\limits_{\varpi^i\in E^i}\gamma^i(\varpi^i).
		\]
		Thus, all original perturbation constraints are satisfied, including the automatic constraint for $q=0$. Since the realization-plan nonnegativity and flow constraints are unchanged, the projection of the extended feasible set is contained in $\prod_{i\in N}\Lambda^i(\delta(\varepsilon))$.
		
		Conversely, let $\gamma\in\prod_{i\in N}\Lambda^i(\delta(\varepsilon))$ and assign each auxiliary variable its corresponding exact value defined above. Each exact minimum is no greater than any feasible sum defining it, so all recursive inequalities hold. The original perturbation constraints and \eqref{eq:fc-root-minimum} also imply the terminal inequalities \eqref{eq:general-terminal}. Together with the original realization-plan constraints, this provides a feasible extension of $\gamma$, proving the reverse inclusion and completing the proof.
	\end{proof}
	
\end{appendices}

\bibliographystyle{elsarticle-harv}
\bibliography{library}

\end{document}